\documentclass{article}

\usepackage{fullpage}

\usepackage[english]{babel}
\usepackage[utf8]{inputenc}

\usepackage[affil-it]{authblk}
\usepackage{amsmath,amsthm,amssymb}
\usepackage{hyperref}

\newtheorem{theorem}{Theorem}[section]       
\newtheorem{lemma}[theorem]{Lemma}

\newtheorem{proposition}[theorem]{Proposition}
\newtheorem{definition}[theorem]{Definition}

\usepackage{tikz}
\usetikzlibrary{quantikz}
\usepackage{hhline}
\usepackage{collectbox}


\usepackage{stmaryrd}
\usepackage{ dsfont }
\usepackage{array}
\usepackage{bbold}
\usepackage{float}
\usepackage{multirow}
\usepackage{algorithmic}
\usepackage{amsmath}

\DeclareMathOperator{\cov}{Cov}

\newcommand*{\cA}{\mathcal{A}}

\newcommand*{\cF}{\mathcal{F}}

\newcommand*{\cI}{\mathcal{I}}
\newcommand*{\cN}{\mathcal{N}}
\newcommand*{\cM}{\mathcal{M}}

\newcommand*{\cO}{\mathcal{O}}
\newcommand*{\cP}{\mathcal{P}}

\newcommand*{\cS}{\mathbf{S}}

\DeclareMathOperator*{\unif}{Unif} 

\DeclareMathOperator{\TV}{TV}

\newcommand*{\eps}{\varepsilon}

\newcommand*{\id}{\mathrm{id}}
\newcommand*{\tr}{\mathrm{Tr}}
\newcommand*{\spr}[2]{\langle #1 | #2 \rangle}

\newcommand*{\pr}[1]{\mathds{P}\left(#1 \right)}
\newcommand*{\ex}[1]{\mathds{E}\left(#1 \right)}

\newcommand{\ceil}[1]	{\left\lceil #1 \right\rceil}

\usepackage{graphicx}
\usepackage{cuted}

\usepackage[maxnames = 5, backend=biber, style=numeric-comp,sorting=none]{biblatex} 
\title{Lower Bounds on Learning Pauli Channels with  Individual Measurements }

\author{Omar Fawzi}
\author{Aadil Oufkir}
\author{Daniel Stilck França} 

\affil{\small{Univ Lyon, Inria, ENS Lyon, UCBL, LIP, Lyon, France}}

\date{\today}

\begin{document}
\maketitle
\begin{abstract} 
 Understanding the noise affecting a quantum device is of fundamental importance for scaling quantum technologies. A particularly important class of noise models is that of Pauli channels, as randomized compiling techniques can effectively bring any quantum channel to this form and are significantly more structured than general quantum channels.
In this paper, we show fundamental lower bounds on the sample complexity for learning Pauli channels in diamond norm. 
We consider strategies that may not use auxiliary systems entangled with the input to the unknown channel and have to perform a measurement before reusing the channel. 
For non-adaptive algorithms, we show a lower bound of $\Omega(2^{3n}\eps^{-2})$ to learn an $n$-qubit Pauli channel. In particular, this shows that the recently introduced learning procedure by \cite{flammia2020efficient} 
is essentially optimal. In the adaptive setting, we show a lower bound of $\Omega(2^{2.5n}\eps^{-2})$ for $\eps=\mathcal{O}(2^{-n})$, and a lower bound of $\Omega(2^{2n}\eps^{-2} )$ for any $\eps > 0$. This last lower bound holds even in a stronger model where in each step, before performing the measurement, the unknown channel may be used arbitrarily many times sequentially interspersed with unital operations.
\end{abstract}

\section{Introduction}

In spite of their impressive progress over the last few years~\cite{Arute2019,Zhong2020,Scholl2021,Ebadi2021}, the scaling and effective employment of quantum technologies still face many challenges. One of the most significant ones is how to tame the noise affecting such devices. For that, more effective tools are required to characterize and learn noisy quantum channels~\cite{Eisert2020}. As the number of parameters required to describe a  quantum channel scales exponentially in the size of the device, it is challenging to learn the noise beyond a few qubits. 

A class of quantum channels that deserves particular attention is that of Pauli channels~\cite[Sec. 4.1.2]{watrous2018theory}. 
In fact, the set of Pauli channels provides a simple and effective model of incoherent noise, admitting a representation in terms of a probability distribution corresponding to different Pauli errors and inheriting the rich structure of the Pauli matrices. In addition, it defines a physically relevant noise model (see e.g., \cite{harper2020efficient}) and the noise affecting a device can always be mapped into a Pauli channel by using randomized compiling~\cite{wallman2016noise} techniques without incurring a loss in fidelity. These properties make the problem of Pauli tomography, i.e., learning a Pauli channel, particularly relevant.

A popular and widely used technique to learn Pauli channels is 
randomized benchmarking and its variations~\cite{flammia2020efficient,magesan2012characterizing,francca2018approximate,helsen2022general,helsen2019new}. The reason for its popularity is that it satisfies several desirable properties: It is robust to errors both in the initial state preparation and measurements (SPAM errors), it only uses simple input states and measurements and it does not use any  auxiliary systems. The motivating question for this work is to ask whether this protocol is optimal given such properties and restrictions or if we can hope to find more efficient protocols. Thus, we derive lower bounds for protocols that learn Pauli channels that fit the setting of randomized benchmarking protocols, in the sense that we are allowed to apply a channel multiple times to an initial state and intersperse its use with unitaries (or more generally unital channels) before measuring the state. However, we are not going to consider protocols where we can perform entangled measurements on multiple outputs of the channel at the same time  or perform measurements on auxiliary systems that are entangled with the system the Pauli channel acts on~\footnote{we will allow the learner to intersperse arbitrary unital operations between uses of the channel. Strictly speaking, it is necessary to act with a unitary on a auxiliary system to implement arbitrary unital operations on a system~\cite{Mendl2009}. Note, however, that we do not need to measure the auxiliary system to obtain an arbitrary unital map, so this operation does not violate our assumptions. However, the reason we add arbitrary unital maps is for the sake of generality, our goal is to capture the more restricted setting of the application of arbitrary unitaries between the channels, which corresponds to the setting of randomized benchmarking.}. 
We refer to this class of measurements as individual measurements. Although allowing for auxiliary systems that we can measure would lead to significantly more efficient and simple protocols as we will discuss later, in practical scenarios it is unclear if it is reasonable to assume we can noiselessly entangle the qubits of a noisy device we wish to characterize to another set of qubits, potentially of the same size as the original device. Thus, even though more complicated, our setup comes from a firm practical motivation.

\paragraph{Contributions}

Given the previous discussion, we introduce a class of strategies, that we call strategies with individual measurements, that encompass some of the desirable properties we mentioned. Denote by $\mathcal{P}$  the unknown $n$-qubit Pauli channel we want to learn. In a strategy with individual measurements, the learner repeats for $t=1, \dots, N$ the following operations: choose an arbitrary $n$-qubit state $\rho_t$, apply the unknown channel $\mathcal{P}$ and then perform a measurement $\mathcal{M}_t$ of his choosing obtaining an outcome $I_t$. The estimate that is returned by the learner is then a function of $I_1, \dots, I_N$. This model is illustrated in Figure~\ref{fig:incoherent-strategies} and  captures the requirement that no auxiliary systems are allowed. Note that such strategies need not be robust and can in principle prepare states and measurements that are not simple, but as we are establishing lower bounds on the resources needed, it only makes the result stronger to allow for more strategies. We refer the reader to Section~\ref{sec:preliminaries} for a formal description of this model. {Strategies with individual measurements for Pauli channels are in direct analogy with state tomography results that do not resort to a quantum memory~\cite{chen2022tight}, a widely studied setting that, as it is the case here for Pauli settings, is motivated by practical limitations of quantum devices.}

\begin{figure}
    \centering
 \begin{quantikz}[thin lines] 
					\gategroup[wires=1,steps=4,style={rounded corners,fill=blue!20,inner xsep=9pt},background]{}\fcolorbox{black}{green!50}{$\rho_{1}$} & \gate[style={fill=red!40}]{\cP} \qw& \gate[1, style={fill=green!50}]{\cM_1} \qw   &\cw\!\rstick{\hspace{-0.2em}$I_1$}  
				\end{quantikz} 
\\\vspace{0.4em}
$\big\downarrow$ 
\\
 \begin{quantikz}[thin lines] 
					\gategroup[wires=1,steps=4,style={rounded corners,fill=blue!20,inner xsep=9pt},background]{}\fcolorbox{black}{green!50}{$\rho_{2}$} & \gate[style={fill=red!40}]{\cP} \qw& \gate[1, style={fill=green!50}]{\cM_2} \qw   &\cw\!\rstick{\hspace{-0.2em}$I_2$} 
				\end{quantikz} 
\\
$\big\downarrow \vdots$
\\
 \begin{quantikz}[thin lines] 
					\gategroup[wires=1,steps=4,style={rounded corners,fill=blue!20,inner xsep=10pt},background]{}\fcolorbox{black}{green!50}{$\rho_{N}$} & \gate[style={fill=red!40}]{\cP} \qw& \gate[1, style={fill=green!50}]{\cM_N} \qw   &\cw\rstick{\hspace{-0.3em}$I_N$} 
				\end{quantikz} 
    \caption{Illustration of {a strategy with individual measurements}. The estimated channel $\widetilde{\cP}$ is computed from $I_1, \dots, I_N$.}
\label{fig:incoherent-strategies}
\end{figure}

We provide lower bounds on the number $N$ of times $\mathcal{P}$ is used by any {strategy with individual measurements} that learns an estimate of $\mathcal{P}$ to precision $\eps > 0$ in the diamond norm. More specifically, writing $d=2^n$, we obtain the following results (summarized in Table~\ref{tab}):

\begin{itemize}
    \item We start by considering non-adaptive strategies, for which the choices $\rho_t$ and $\cM_t$ do not depend on the previous measurement outcomes $I_1, \dots, I_{t-1}$. We show that any non-adaptive  {strategy  with individual measurements} for Pauli channels has to use the channel $N \geq \Omega(d^3/\eps^2)$ times. In particular, this shows that the randomized benchmarking algorithm of \cite[Result $1$]{flammia2020efficient} is optimal since the channels we consider in our construction have a spectral gap $\Delta\ge 1-4\eps$ and thus the total number of channel uses is at most twice the number of measurements. 
    This result is stated in Theorem~\ref{thm:LBNA}. 
    For the proof, we follow similar strategies pursued for learning quantum states \cite{paninski2008coincidence,flammia2012quantum,haah2016sample,lowe2022lower} and  construct an $\eps$-separated family of Pauli channels close to the maximally depolarizing channel and use it to encode a message from $[e^{\Omega(d^2)}]$ {through the sequence of outcomes that the learning algorithm produces. The correctness of the learning algorithm ensures that this   message is  decoded correctly with the same success probability.} Hence, the encoder and decoder should share at least $\Omega(d^2)$ nats of information. On the other hand, after each step, we show that the correlation between the encoder and decoder can only increase by at most $\cO(\eps^2/d)$ nats each time the channel is used.  
    Note that the naive upper bound on this correlation is $\cO(\eps^2)$, we obtain an improvement by a factor $d$ by exploiting the randomness in the construction of the Pauli channel. Our result holds in a stronger model where the channel $\cP$ can be used $m_t$ times intertwined with arbitrary unital channels before performing the measurement $\cM_t$. In this case, the condition satisfied by any algorithm is $N \ge  \Omega(d^3/\eps^2)$ or $\sum_{t=1}^N m_t \ge  \Omega(d^4/\eps^6)$.
    \item 
    In the more general adaptive setting, we first show that any {strategy with individual measurements} for Pauli channels should satisfy $N \geq \Omega(d^2/\eps^2)$. This bound holds in the stronger model where the unknown channel $\cP$ can be used $m_t$ times before the measurement, and the bound does not depend on $m_t$.
        This result is stated in Theorem~\ref{thm: GLB}. Furthermore, our main result about adaptive strategies is a lower bound $N \geq \Omega(d^{2.5}/\eps^2)$ provided $\eps\le 1/(20d)$. This result is stated in Theorem~\ref{thm: LB-adaptive}. The structure of the proof is similar to the one for non-adaptive setting but for adaptive strategies, it is more complicated to bound the increase in the information we obtain at each step of the algorithm. 
    For this, we change the previous construction and use normalized Gaussian random variables in the Pauli channel's coefficients.  The Gaussian variables allow us to break the dependency between the probability of measurements at different steps by applying Gaussian integration by parts on an upper bound of the mutual information. With this, we show that the information that is obtained by each new step is at most $\cO(k\eps^4/d^3)$ nats at step $k$ which gives the claimed bound.
\end{itemize}
\begin{table}[t]
\centering
\caption{Lower  and upper bounds for Pauli channel tomography using individual measurements. $N$ is the total number of steps or   measurements. 
}
\begin{tabular}{  |c|c| c|} 
\hline
  \textbf{Model}  &\textbf{Lower bound} &  \textbf{Upper bound}\\
  \hline\hline
    Non-adaptive   & $N \ge  \Omega(d^3/\eps^2)$  & $N=\tilde{\cO}(d^3/\eps^2)$\\
   $\ell_1$-distance&  [this work]&   \cite{flammia2020efficient}   \\ 
   \hline
  Non-adaptive & $N \ge  \Omega(1/\eps^2)$ &$N=\tilde{\cO}(1/\eps^2)$\\
 $\ell_\infty$-distance &   \cite{flammia2021pauli}& \cite{flammia2021pauli}\\
   \hline
  Adaptive &  $N\ge \Omega(d^2/\eps^2)$  &$N=\tilde{\cO}(d^3/\eps^2)$\\
  $\ell_1$-distance&[this work]&\cite{flammia2020efficient}\\
  \hline
   Adaptive,  $\ell_1$-distance  &$N\ge \Omega(d^{2.5}/\eps^2)$   & $N=\tilde{\cO}(d^3/\eps^2)$     \\
    $\eps\le 1/(20d)$ & [this work] &\cite{flammia2020efficient} \\
    \hline 
\end{tabular}
\label{tab}
\end{table}

\paragraph{Related work} 
Learning Pauli channels has been considered in different settings. \cite{flammia2020efficient} provides an algorithm for learning Pauli channels in $\ell_2$-norm using $\Tilde{\cO}( d/\eps^2)$ measurements. This implies an upper bound of $\Tilde{\cO}( d^3/\eps^2)$ for learning Pauli channel in $\ell_1$-norm. 
This article addresses an open question posed in \cite{flammia2020efficient} about a lower bound for learning Pauli channels. As previously discussed, we show that the algorithm of \cite{flammia2020efficient} is optimal up to logarithmic factors.
Moreover, learning a Pauli channel in $\ell_\infty$-norm was shown to be solvable with $\tilde{\Theta}(1/\eps^2)$ measurements in \cite{flammia2021pauli} and this is optimal up to logarithmic factors. 
These algorithms do not use an ancilla system.
The work of \cite{chen2022quantum} shows an exponential separation between allowing and not allowing ancilla for estimating the Pauli eigenvalues in $\ell_\infty$-norm. Using the Parseval–Plancherel identity, their upper bound can be translated to learning in $\ell_1$-norm with an $n$-qubit ancilla assisted algorithm using $\tilde{\cO}(d^2/\eps^2)$  measurements. {The upper bound ${\cO}(d^2/\eps^2)$ for the $\ell_1$-norm can be achieved through Bell sampling \cite{flammia2021pauli}.} However, our lower bounds do not apply in this setting since we only consider ancilla-free strategies. We also note that~\cite{chen2022quantum} shows a lower bound of $\Omega(d^{\frac{1}{3}}/\eps^{2})$ { individual measurements}  to learn the eigenvalues of $\mathcal{P}$ in the adaptive  setting up to  $\eps$ in $\ell_\infty$-norm and $\Omega(d/\eps^{2})$ in the non-adaptive  setting. 
However, this is a different figure of merit than the one we consider. Moreover, one could argue that the diamond norm, which to the best of our knowledge was not considered before this work, provides the strongest and operationally motivated definition of learning a channel, as it is again in direct analogy with the trace distance for states~\cite{watrous2018theory,chen2022tight}.

Other noteworthy protocols to learn quantum channels include gate set tomography~\cite{gate_set_tomo} and techniques based on compressed sensing~\cite{PhysRevLett.121.170502}. Although they apply to more general classes of channels, they do not offer quantitative or qualitative advantages over randomized benchmarking in the setting of Pauli channels. We refer the readers to the survey \cite{montanaro2013survey} for results on testing quantum channels and to~\cite{Eisert2020} for quantum channel learning. 
For the state tomography problem, it is known that the optimal copy complexity for {strategies with individual measurements} in both adaptive and non-adaptive settings is $\Theta(d^3/\eps^2)$ \cite{haah2016sample,chen2022tight}. 
 In contrast, for  non-adaptive {strategies with individual measurements}, quantum channel tomography in the diamond norm can be done using $\Tilde{\Theta}\left(d^6/\eps^2\right)$ copies \cite{surawy2022projected,oufkir-sample-optimal}. 
However, if we add the Pauli structure to the channel, our lower bound along with the upper bound of \cite{flammia2020efficient} show that the optimal complexity is the same as state tomography complexity.

\section{Preliminaries} 
\label{sec:preliminaries}

Let $ \mathds{N}^*$  be the set of positive integers and $n\in \mathds{N}^*$. Let $d  \coloneqq  2^n$ be the dimension of an $n$-qubit system. We use the notation $[d] \coloneqq  \{1,\dots,d\}$. We adopt the bra-ket notation: a column vector is denoted $\ket{\phi}$ and its adjoint is denoted $\bra{\phi}=\ket{\phi}^\dagger$. With this notation, $\spr{\phi}{\psi}$ is the dot product  of the vectors $\phi$ and $\psi$ and, for a unit vector $\ket{\phi}$,  $\proj{\phi}$ is the rank-$1$ projector on the space spanned by the vector $\phi$. The set of unit vectors is denoted $\cS^d\coloneqq\{\ket{\phi} \in \mathds{C}^{d} : \spr{\phi}{\phi}=1 \}$. The canonical basis $\{e_i\}_{i\in [d]}$ is denoted $\{\ket{i}\}_{i\in [d]} \coloneqq  \{\ket{e_i}\}_{i\in [d]}$. A quantum state is a positive semi-definite Hermitian matrix of trace $1$. We will denote the identity matrix by $\mathds{I}_d\in\mathds{C}^{d\times d}$ and by $\id_d: \mathds{C}^{d\times d}\rightarrow \mathds{C}^{d\times d}$ the identity map. We will omit the the $d$ subscript if the dimension is clear from context. A quantum channel is a map $\cN: \mathds{C}^{d\times d}\rightarrow \mathds{C}^{d\times d}$ of the form $\cN(\rho)=\sum_{k}A_k \rho A_k^\dagger$ where the Kraus operators $\{A_k\}_k$ satisfy $\sum_k A_k^\dagger A_k=\mathds{I}$. 
A map $\cN$ is a quantum channel if, and only if, it is:
\begin{itemize}
    \item \textbf{completely positive:} for all $\rho\in \mathds{C}^{d^2\times d^2},\rho\succcurlyeq 0$, $[\id_d\otimes \cN](\rho)\succcurlyeq 0$ and
    \item \textbf{trace preserving:} for all $\rho\in \mathds{C}^{d\times d}$, $\tr(\cN(\rho))=\tr(\rho)$.
\end{itemize}
If the quantum channel $\cN$ satisfies further $\cN(\mathds{I})=\mathds{I}$, it is called \textit{unital}.

We define the diamond distance between two quantum channels $\cN$ and $\cM$ as the diamond norm of their difference:
\begin{align*}
    \|\cN-\cM\|_\diamond \coloneqq  \max_{\phi: \spr{\phi}{\phi}=1}\|[\id_d\otimes(\cN-\cM)](\proj{\phi}) \|_1
\end{align*}
where the the Schatten $1$-norm of a matrix $M$ is defined as $\|M\|_1 \coloneqq \tr\left(|M|\right)$ and $|M|\coloneqq \sqrt{M^\dagger M}$.

Pauli channels are quantum channels whose Kraus operators are weighted Pauli operators. Formally, an $n$-qubit Pauli channel $\cP$ can be written as follows:
\begin{align}\label{equ:pauli_channel}
    \cP(\rho)= \sum_{P\in \{\mathds{I},X,Y,Z\}^{\otimes n} } p(P) P \rho P
\end{align}
where the Pauli matrices \[\mathds{I}=\begin{pmatrix}
1 & 0 \\
0& 1
\end{pmatrix}, X=\begin{pmatrix}
0 & 1 \\
1& 0
\end{pmatrix}, Y=\begin{pmatrix}
0 & - \mathrm{i} \\
\mathrm{i} & 0
\end{pmatrix}, Z=\begin{pmatrix}
1 & 0 \\
0 & -1
\end{pmatrix} \] and $\{p(P)\}_{P\in \{\mathds{I},X,Y,Z\}^{\otimes n}}$ is a probability distribution. 
Let $\mathds{P}_n= \{\mathds{I},X,Y,Z\}^{\otimes n}$ be the set of Pauli operators. The elements of $\mathds{P}_n$ either commute or anti commute. Let $P$ and $Q$ be two Pauli operators, we have $PQ=(-1)^{P {\circ} Q}QP$ where $P{\circ}Q=0$ if $[P,Q]=0$ and $P{\circ}Q=1$ otherwise. 

We consider the Pauli channel tomography problem which consists of learning a Pauli channel in the diamond norm. Given a precision  parameter $\eps>0$, the goal is to construct a  Pauli channel $\widetilde{\cP}$ satisfying with at least a probability $2/3$:
\begin{align*}
    \|\cP-\widetilde{\cP}\|_\diamond\le \eps.
\end{align*}
An algorithm $\cA$ takes as input {$n$, is given a black-box (a.k.a.\ oracle) implementation of an unknown Pauli channel }
 $\cP$ and outputs a classical description of a Pauli channel $\widetilde{\cP}$. The
 algorithm $\cA$ is $1/3$-correct for this problem if it outputs a  Pauli channel $\widetilde{\cP}$ that is  $\eps$-close to $\cP$ with a probability of error at most $1/3$.
We choose to learn in the diamond norm because it characterizes the minimal error probability to distinguish between two quantum channels  when auxiliary systems are allowed \cite{watrous2018theory}. Since the diamond norm between two Pauli channels is exactly twice the $\TV$-distance between their corresponding probability distributions \cite{magesan2012characterizing}, approximating the Pauli channel $\cP$ in diamond norm is equivalent to approximating the probability distribution $p$ in $\TV$-distance. The latter is defined for two probability distributions $p$ and $q $ on $[d]$ as follows:
\begin{align*}
    \TV(p,q) \coloneqq  \frac{1}{2}\sum_{i=1}^d |p_i-q_i|.
\end{align*}

The learner can only extract classical information from the unknown $n$-qubit Pauli channel $\cP$ by performing a measurement on the output state. Throughout the paper, we only consider unentangled or {individual}  measurements. That is, the learner can only measure with an $n$-qubit measurement device and auxiliary qubits or {processing entangled multiple copies of the unknwon channel at once (i.e., $\cP^{\otimes n}$ for $n\in \mathds{N}^*$)}   are not allowed. This restriction is natural for the problem at hand, given that performing measurements on multiple copies requires a quantum memory, which is currently unavailable in the vast majority of experimental platforms.

More precisely, an $n$-qubit measurement is defined by a POVM (positive operator-valued measure) with a finite number of elements: this is a set of positive semi-definite matrices $\mathcal{M}=(M_i)_i $ acting on the Hilbert space $\mathds{C}^{2^n}$ and satisfying $\sum_i M_i=\mathds{I}$. The element $M_i$ in the POVM $\mathcal{M}$ is associated with the outcome $i$. The tuple $(\tr(\rho M_i))_i$ is non-negative and sums to $1$: it thus defines a probability distribution. Born's rule \cite{1926ZPhy...37..863B} says that the probability that the measurement on a quantum state $\rho$ using the POVM $\cM$ will output $i$ is exactly $\tr(\rho M_i)$. 

For an integer $t\ge 1$, we say that the learner is at step $t$ if it has already performed $t-1$ measurements. With this definition, the total number of steps is exactly the total number of measurements. However, depending on the setting, the total number of channel uses could be different than the total number of steps. The goal of the paper is to show lower bounds on the total number of steps as well as the total number of the channel uses.

A simple example we can propose to see the effect of reusing the channel is the following test: $H_0: \cP(\rho)=\rho$ vs $H_1: \cP(\rho)= (1-\eps)\rho+ \eps\tr(\rho)\frac{\mathds{I}}{d}$. We can choose as input the rank one state $\rho= \proj{0}$. Under the null hypothesis $H_0$, the channel does not affect the state $\proj{0}$. On the other hand, under $H_1$, if we apply the channel $\cP$ a number  $m\in \mathds{N}^* $ times the resulting quantum state is $\cP^{(m)}(\rho)=(1-\eps)^m\proj{0}+\left(1-(1-\eps)^m\right)\frac{\mathds{I}}{d}$. Hence, if we measure with the POVM $\cM= \{\proj{0}, \mathds{I}-\proj{0}\}$ of outcomes $0$ and $1$ respectively, under $H_0$ we will always see $0$ while under $H_1$, we will see $0$ with probability roughly $(1-\eps)^m$. Therefore, we can achieve a probability of error at most $\delta$ with only \textit{one measurement} but the channel is reused $\log(1/\delta)/\eps$-times\footnote{all the logs are taken in base $\mathrm{e}$ so the information is measured in ‘‘nats''.}. However, if we do not allow the application of the channel to the same register, then the number of measurements needed is  approximately  $\log(1/\delta)/\eps$.

For a tuple $I = (I_1, \dots, I_N)$, we denote, for $k \in [N]$, $I_{\leq k} = I_{< k+1} = (I_1, \dots, I_k)$. 


\section{A general lower bound on the number of steps required for Pauli channel tomography}

In this section, we consider the problem of learning a Pauli quantum channel using {individual} measurements. Unlike the usual state tomography problem for which at each step the learner can only choose the measurement device, for quantum channels, the learner has additional choices. First, in every setting, the learner can choose the input quantum state at each step. This choice can be done in an adaptive fashion: the input quantum state at a given step can be chosen depending on the previous observations (and of course the previous input states and POVMs). Second, the learner has the ability to reuse the Pauli quantum channel as much as it wants before performing the measurement. This is specific to quantum process tomography too since for state tomography using {individual} measurements, once a measurement is performed, the post-measurement quantum state is usually useless. Finally, the learner can intertwine arbitrary unital quantum channels and the unknown Pauli quantum channel before measuring the output of this (possibly long) sequence of quantum channels. 
We propose a lower bound on the number of steps required for the Pauli channel tomography problem in this general setting. 

Recall that  Pauli channel tomography problem is equivalent to  learning the probability $p$ in the $\TV$-distance. 
\begin{figure*}
    \centering
\begin{quantikz}[thin lines] 
\gategroup[wires=1, steps=10,style={rounded corners,fill=blue!20,inner xsep=13pt},background]{}\lstick{$\rho_1$} & \gate[style={fill=red!40}]{\cP} \qw  & \gate[style={fill=green!50}]{\cN_1}\qw & \gate[style={fill=red!40}]{\cP} \qw &\qw~~\dots~~ & \gate[style={fill=red!40}]{\cP} \qw  &  \gate[style={fill=green!50}]{\cN_{m_1-1}} \qw& \gate[style={fill=red!40}]{\cP}\qw & \gate[style={fill=green!50}]{\cM_{1}}   &\cw \rstick{\hspace{-0.2em}$I_1$}  
\end{quantikz}
\\
\vspace{0.22em}
$\big\downarrow (I_1)$
\\
\begin{quantikz}[thin lines] 
\gategroup[wires=1,steps=10,style={rounded corners,fill=blue!20,inner xsep=14pt},background]{}\lstick{$\rho_{2}^{I_{1}}$} & \gate[style={fill=red!40}]{\cP} \qw  & \gate[style={fill=green!50}]{\cN_1^{I_{1}}}\qw & \gate[style={fill=red!40}]{\cP} \qw &\qw~~\dots~~ & \gate[style={fill=red!40}]{\cP} \qw  &  \gate[style={fill=green!50}]{\cN_{m_2-1}^{I_{1}}} \qw& \gate[style={fill=red!40}]{\cP}\qw & \gate[style={fill=green!50}]{\cM_{2}^{I_{1}}}   & \rstick{\hspace{-0.2em}$I_2$} \cw
\end{quantikz} 
\\
$\big\downarrow \vdots$
\\
\begin{quantikz}[thin lines] 
\gategroup[wires=1,steps=10,style={rounded corners,fill=blue!20,inner xsep=33pt},background]{}\lstick{$\rho_{N-1}^{I_{<N-1}}$\hspace{-0.2em}} & \gate[style={fill=red!40}]{\cP} \qw  & \gate[style={fill=green!50}]{\cN_1^{I_{<N-1}}}\qw & \gate[style={fill=red!40}]{\cP} \qw &\qw~~\dots~~ & \gate[style={fill=red!40}]{\cP} \qw  &  \gate[style={fill=green!50}]{\cN_{m_{N-1}-1}^{I_{<N-1}}} \qw& \gate[style={fill=red!40}]{\cP}\qw & \gate[style={fill=green!50}]{\cM_{N-1}^{I_{<N-1}}}   & \rstick{\hspace{-0.2em}$I_{N-1}$} \cw
\end{quantikz}
\\\vspace{0.22em}
$\big\downarrow I_{<N}=(I_1,\dots, I_{N-1})$
\\
\begin{quantikz}[thin lines] 
\gategroup[wires=1,steps=10,style={rounded corners,fill=blue!20,inner xsep=29pt},background]{}\lstick{$\rho_{N}^{I_{<N}}$\hspace{-0.2em}} & \gate[style={fill=red!40}]{\cP} \qw  & \gate[style={fill=green!50}]{\cN_1^{I_{<N}}}\qw & \gate[style={fill=red!40}]{\cP} \qw &\qw~~\dots~~ & \gate[style={fill=red!40}]{\cP} \qw  &  \gate[style={fill=green!50}]{\cN_{m_{N}-1}^{I_{<N}}} \qw& \gate[style={fill=red!40}]{\cP}\qw & \gate[style={fill=green!50}]{\cM_{N}^{I_{<N}}}   & \rstick{\hspace{-0.2em}$I_{N}$} \cw
\end{quantikz}
    \caption{Illustration of an adaptive strategy for learning Pauli channel.}
    \label{fig: adaptive }
\end{figure*}
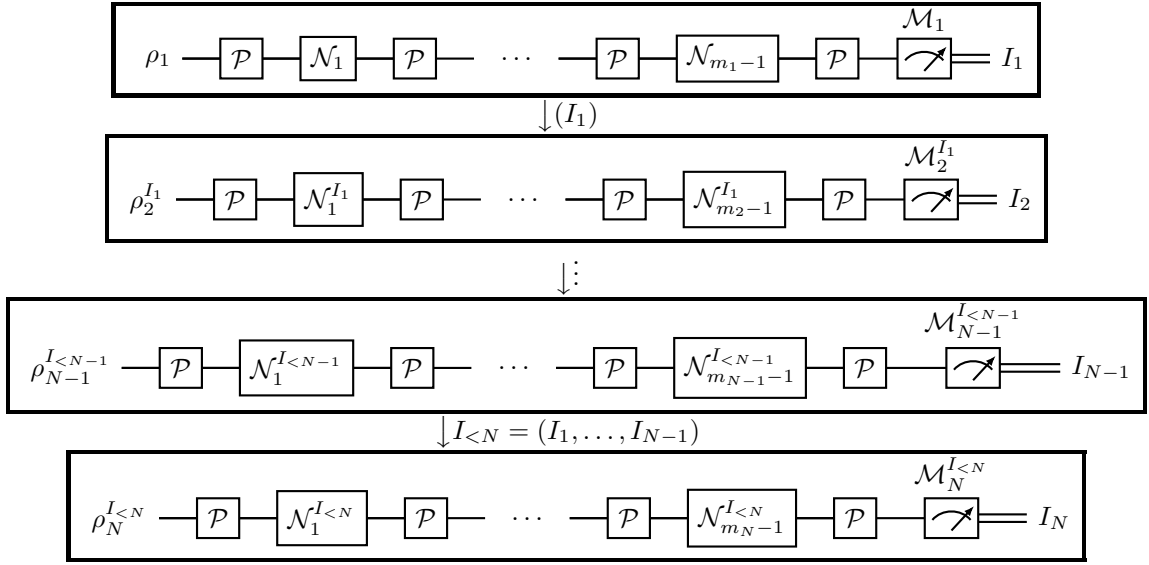

\begin{definition}Let $\cP$ be a Pauli channel and let $N$ be a sufficient number of steps to learn $\cP$ as defined in \eqref{equ:pauli_channel}.  
    At step $t\in [N]$, an adaptive {strategy with individual measurements} has the ability to choose an input quantum state $\rho_t$, the number $m_t\ge 1$ of uses of the quantum channel $\cP$, the unital quantum channels applied in between $\cN_1,\dots,\cN_{m_t-1}$ and the POVM $\cM_t$ for measuring the output quantum state $\rho^{\text{output}}_t$:
\begin{align*}
    \rho^{\text{output}}_t= \underbrace{\cP \circ \cN_{m_t-1}\circ\cP \circ \dots \circ \cP \circ \cN_{1} \circ \cP (\rho_t)}_{\cP \text{ is applied } m_t \text{ times}}.
\end{align*}
All these elements can be chosen adaptively: the choice of $m_t,\rho_t, \cN_1,\dots,\cN_{m_t-1}$ and $\cM_t$ can depend on the previous observations $I_1,\dots, I_{t-1}$ (see Fig.~\ref{fig: adaptive } for an illustration). However, to not overload the expressions we do not add the subscript $I_1,\dots, I_{t-1}$ on $m_t, \rho_t, \cN_1,\dots,\cN_{m_t-1}$ or $\cM_t$. By Born's rule, performing a measurement on the output quantum state $\rho^{\text{output}}_t$ using the POVM $\cM_t=\{M^t_i\}_{i\in \cI}$ is equivalent to sampling from the probability distribution 
\begin{align*}
    x_t \sim  \{\tr(\rho^{\text{output}}_tM^t_i)\}_{i\in \cI}.
\end{align*}
The observations $(x_1, \dots, x_N)$ are used to construct a probability distribution $\hat{p}$ on the set of Pauli operators $\mathds{P}_n$ satisfying with a probability at least  $2/3$:
\begin{align*}
    \TV(p, \hat{p}) \le \eps.
\end{align*}
\end{definition}
Note that unital operations cannot be used to prepare a new state and thus have a free step. 
In fact, applying a unital operation after a noisy Pauli channel cannot prepare a rank-$1$ state for example.
We propose the following lower bound on the number of steps $N$.
\begin{theorem}\label{thm: GLB}
The problem of Pauli channel tomography using  ancilla-free {individual} measurements requires a number of steps satisfying:
\begin{align*}
    N\ge \Omega\left( \frac{d^2}{\eps^2}\right).
\end{align*}
\end{theorem}

This theorem shows that no matter how often the learner reuses the quantum Pauli channel intertwined with other unital quantum channels on each step, the global number of steps should be exponential in the number of qubits. This can be explained by the fact that a Pauli channel adds noise to the input state, so reapplying it makes the input state more noisy and can't help to extract more information. Although, as we remark later, this lower bound is weaker in the dependency on the dimension $d$ compared to the non-adaptive case, it has the particularity of not depending on the number of uses of the Pauli channel. 
\begin{proof}\label{app: proof of thm glb}
We will break down the proof into several steps, which we outline below:
\paragraph{Construction of the family $\cF$}\label{app:construction}
We start by describing a general construction of a big family $\cF=\{\cP_x\}_{x\in[M]}$ constituted of  quantum Pauli channels satisfying for all $x\neq y \in [M]: \TV(p_x,p_y)\ge \eps$, we say that the family $\cF$ is $\eps$-separated. These quantum channels have the form for $x\in [M]$: 
\begin{align}\label{form}
    \cP_x(\rho)&= \sum_{P\in \{\mathds{I},X,Y,Z\}^{\otimes n} } p_x(P) P \rho P\notag
    \\&=\sum_{P\in \{\mathds{I},X,Y,Z\}^{\otimes n} } \left(\frac{1+4\alpha_x(P)\eps}{d^2} \right) P \rho P
\end{align}
where $\alpha_x(P)=\pm 1$ are chosen randomly so that  $\alpha_x(P)=-\alpha_x(\sigma(P))$ for some perfect matching $\sigma$ of $\{\mathds{I},X,Y,Z\}^{\otimes n}$. The latter condition ensures  $\sum_{P\in \mathds{P}_n} \alpha_x(P)=0$. Hence, $\cP_x$  is a valid quantum channel for $\eps\le 1/4$.  

Suppose that we have already constructed an $\eps$-separated  family of Pauli quantum channels $\cF=\{\cP_x\}_{x\in[M]}$ of cardinality $M$. We show that we can add another element to this family as long as $M<e^{cd^2}$ for some sufficiently small constant $c$. For this, we choose $\alpha(P)=-\alpha(\sigma(P))=\pm 1$ with probability 1/2, independently for each edge $\{P, \sigma(P)\}$ in the matching. 
This $\alpha$ leads to a quantum channel $\cP(\rho)=\sum_{P\in \{\mathds{I},X,Y,Z\}^{\otimes n} } \left(\frac{1+4\alpha(P)\eps}{d^2} \right) P \rho P$.  
Then, we control the probability that the corresponding Pauli quantum channel isn't $\eps$-far from the family $\cF$. We denote the set of edges in the matching $\sigma$ by $\mathds{P}_n/\sigma$. 
By the union bound and Chernoff-Hoeffding inequality \cite{hoeff}:

\begin{align*}
    \pr{\exists \cP_x\in \cF :  \TV(p,p_x) <\eps }
    &\le\sum_{x=1}^M  \pr{ \sum_{P\in \mathds{P}_n } |p(P)-p_x(P)|< 2\eps }
      =
     \sum_{x=1}^M  \pr{ \sum_{P\in \mathds{P}_n } 4|\alpha(P)-\alpha_x(P)|< 2d^2}
    \\& =  \sum_{x=1}^M  \pr{ \sum_{P\in \mathds{P}_n} \mathbb{1}_{\alpha(P)\neq \alpha_x(P)}< \frac{d^2}{4}}
    =
    \sum_{x=1}^M  \pr{ \frac{2}{d^2}\sum_{ \{P, \sigma(P)\}\in \mathds{P}_n/\sigma}  \mathbb{1}_{\alpha(P)\neq \alpha_x(P)}  < \frac{1}{4}}
    \\&\le  \sum_{x=1}^M  \exp(-2(d^2/2)(1/4)^2  )
    =M \exp(-d^2/16  )
\end{align*}
which is strictly smaller than $1$ if $M<e^{d^2/16}$.  
So far, we have proven the following lemma:
\begin{lemma}
 There exists an $\eps$-separated family $\cF$ of quantum Pauli channels of the form in \eqref{form} and size at least $e^{d^2/16}$.
\end{lemma}
Hence, we can use this family to encode a message  $X\sim \unif\{[M]\}$  to the sequence of  outcomes produced by  the learning algorithm
 when provided with the quantum Pauli channel $\cP= \cP_{X}$.  More precisely, the learning algorithm 
chooses its inputs states and performs {individual} measurements possibly after many uses of the channel $\cP_{X}$ intertwined with arbitrary unital quantum channels, and
observes a sequence of outcomes that will be transmitted to the decoder.
Upon receiving this sequence of outcomes, the decoder runs the data-processing part of the learning algorithm to  produce a Pauli quantum channel $\hat{\cP}$ corresponding to a probability distribution $\hat{p}$ satisfying, with a probability at least $2/3$, $\TV(\hat{p},p_X)\le \eps/2$. Since the family of probability distributions $\{p_x\}_{x\in [M]}$ is $\eps$-separated, there is only one $\hat{X}$ such that $\TV(\hat{p},p_{\hat{X}})\le \eps/2$. 
Therefore a $1/3$-correct algorithm can decode with a probability of failure at most $1/3$. By Fano's inequality, the encoder and decoder should share at least $\Omega(\log(M))\ge \Omega(d^2)$ nats of information.
\begin{lemma}[\cite{fano1961transmission}]\label{inequ-info-lb}
The mutual information between the index of the actual channel $X$  and the estimated index $\hat{X}$ is at least
 \begin{align*}
\cI(X: \hat{X}) \ge (2/3) \log(M) -\log(2)\ge \Omega(d^2).
\end{align*}
\end{lemma}

\paragraph{Upper bound on the mutual information}
We show that no algorithm can extract more than $\cO(\eps^2)$ nats of information at each step. For this, recall that $X$  is the uniform random variable on the set $[M]$ representing the encoder and  denote by $I_1,\dots,I_N$ the {sequence of outcomes produced by the data-acquisition part of the learning algorithm}. The Data-Processing inequality implies:
\begin{align*}
    \cI(X:\hat{X})\le \cI(X:I_1,\dots,I_N).
\end{align*}
Recall the notation $I_{\le k-1} \coloneqq   (I_1,\dots,I_{k-1})$ for all $1\le k\le N$, the chain rule of mutual information gives:
\begin{align*}
    \cI(X:I_1,\dots, I_N) &=\sum_{k=1}^N \cI(X:I_k| I_{\le k-1})
\end{align*}
where $\cI(X:I_k| I_{\le k-1})$ denotes the conditional mutual information between $X$ and $I_k$ giving $I_{\le k-1}$.
We claim that every conditional  mutual information $\cI(X:I_k| I_{\le k-1})$ can be  upper bounded by $\cO(\eps^2)$. To prove this claim, we prove first a general upper bound on the conditional mutual information. 

At step $t\in [N]$ , the $1/3$-correct algorithm used by the decoder chooses the input state $\rho_t$, uses the unknown quantum Pauli channel $\cP$ $m_t\ge 1$ times, eventually intertwines the $\cP$ with unital quantum channels $\cN_1^t, \cN_2^t, \dots, \cN_{m_t-1}^t$ and finally measures the output with a POVM $\cM_t=\{\lambda_i^t \proj{\phi_i^t}\}_{i\in \cI_t}$ where $\spr{\phi_i^t}{\phi_i^t}=1$ and $\sum_i \lambda_i^t \proj{\phi_i^t} = I$. Note that this implies $\sum_i \lambda_i^t = d$. Observe that we can always reduce the measurement with a general POVM $\cM$ to the measurement with such a POVM by taking the projectors on the  eigenvectors of each element of the POVM $\cM$ weighted by the corresponding eigenvalues. 
We denote by $\cP^{m_t}(\rho_t)=\underbrace{\cP \circ \cN^t_{m_t-1}\circ\cP \dots  \cP \circ \cN^t_{1}\circ  \cP (\rho_t)}_{\cP \text{ is applied } m_t \text{ times}}  $ the quantum channel applied to the input quantum state $\rho_t$. 
We denote by $q$ the joint distribution of $(X,I_1,\dots,I_N)$:
\begin{align*}
    q(x,i_1,\dots, i_N)=\frac{1}{M}\prod_{t=1}^N\lambda_{i_t}^t\bra{\phi_{i_t}^t} \cP_x^{m_t}(\rho_t)\ket{\phi_{i_t}^t}.
\end{align*}
We use the usual notation of marginals by omitting the indices on which we marginalize. For instance, for all adaptive algorithms, for all $1\le k\le N$, we have:
\begin{align*}
   q_{\le k}(x,i_1,\dots, i_k)
   &=\sum_{i_{k+1},\dots,i_N}\frac{1}{M}\prod_{t=1}^N\lambda_{i_t}^t\bra{\phi_{i_t}^t} \cP_x^{m_t}(\rho_t)\ket{\phi_{i_t}^t}
   \\&=\!\frac{1}{M}\!\prod_{t=1}^k\!\lambda_{i_t}^t\!\bra{\phi_{i_t}^t} \!\cP_x^{m_t}(\rho_t)\!\ket{\phi_{i_t}^t} \!\prod_{t=k+1}^N\! \sum_{i_t}  \!\lambda_{i_t}^t\!\bra{\phi_{i_t}^t} \!\cP_x^{m_t}(\rho_t)\!\ket{\phi_{i_t}^t}
   \\&=\frac{1}{M}\prod_{t=1}^k\lambda_{i_t}^t\bra{\phi_{i_t}^t} \cP_x^{m_t}(\rho_t)\ket{\phi_{i_t}^t} \prod_{t=k+1}^N \tr(\cP_x^{m_t}(\rho_t))
   \\&= \frac{1}{M}\prod_{t=1}^k\lambda_{i_t}^t\bra{\phi_{i_t}^t} \cP_x^{m_t}(\rho_t)\ket{\phi_{i_t}^t}.
\end{align*}
We sometimes abuse the notation and use $q$ instead of $q_{\le k }$ when it is clear from the context. In order to simplify the expressions, we introduce the notation 
$u_{i_k}^{k,x}= \bra{\phi_{i_k}^k} d\cP_x^{m_k}(\rho_k)-\mathds{I}\ket{\phi_{i_k}^k}$. Note that for adaptive strategies the vectors  $\ket{\phi_{i_k}^k}=\ket{\phi_{i_k}^t(i_{<k})}$ and the states  $\rho_k=\rho_k(i_{<k})$ depend on the previous observations $i_{<k}=(i_1,\dots, i_{k-1})$ for all $k\in [N]$.
Then the general upper bound on the conditional mutual information is:
\begin{lemma}\label{Upper bound on cond mutual info} Let $1\le k \le N$ and $u_{i_k}^{k,x}= \bra{\phi_{i_k}^k(i_{<k})} d\cP_x^{m_k}(\rho_k(i_{<k}))-\mathds{I}\ket{\phi_{i_k}^k(i_{<k})}$. We have for adaptive strategies:
\begin{align*}
    \cI(X:I_k| I_{\le k-1})\le 3\mathds{E}_{x}\mathds{E}_{{i_{\le k-1}}\sim q_{\le k-1}}  \Bigg[\sum_{i_k}\frac{\lambda_{i_k}^k}{d} (u_{i_k}^{k,x})^2\Bigg].
\end{align*}
Moreover, for non-adaptive strategies $u_{i_k}^{k,x}= \bra{\phi_{i_k}^k} d\cP_x^{m_k}(\rho_k)-\mathds{I}\ket{\phi_{i_k}^k}$ and:
\begin{align*}
    \cI(X:I_k| I_{\le k-1})\le 3\mathds{E}_{x} \Bigg[\sum_{i_k}\frac{\lambda_{i_k}^k}{d} (u_{i_k}^{k,x})^2\Bigg].
\end{align*}
\end{lemma}
\begin{proof}[Proof of Lemma~\ref{Upper bound on cond mutual info}]
We can remark that, for all $1\le k\le N$, $q(x,i_{\le k})=\lambda_{i_k}^k\left(\frac{1+u_{i_k}^{k,x}}{d} \right)q(x,i_{\le k-1})$ thus
\begin{align*}
    \frac{q(x,i_k| i_{\le k-1})}{q(x|i_{\le k-1})q(i_k|i_{\le k-1})}
   &=\frac{q(x,i_{\le k})q(i_{\le k-1}) }{q(x,i_{\le k-1})q(i_{\le k})}
    =\frac{\lambda_{i_k}^k\left(\frac{1+u_{i_k}^{k,x}}{d} \right)q(x,i_{\le k-1})q(i_{\le k-1}) }{q(x,i_{\le k-1})\sum_y q(y,i_{\le k})}
    \\&=\frac{\lambda_{i_k}^k\left(\frac{1+u_{i_k}^{k,x}}{d} \right)q(i_{\le k-1}) }{\sum_y q(y,i_{\le k})}
     =\frac{\lambda_{i_k}^k\left(\frac{1+u_{i_k}^{k,x}}{d} \right)q(i_{\le k-1}) }{\sum_y q(y,i_{\le k-1})\lambda_{i_k}^k\left(\frac{1+ u_{i_k}^{k,y}}{d} \right)}
    \\& =\frac{(1+u_{i_k}^{k,x}) q(i_{\le k-1}) }{\sum_y q(y,i_{\le k-1})(1+ u_{i_k}^{k,y})}
     =\frac{(1+u_{i_k}^{k,x}) }{\sum_y q(y|i_{\le k-1})(1+ u_{i_k}^{k,y})}.
\end{align*}
Therefore by Jensen's inequality:
\begin{align*}
    \cI(X:I_k| I_{\le k-1})
    &=\ex{\log\left( \frac{q(x,i_k| i_{\le k-1})}{q(x|i_{\le k-1})q(i_k|i_{\le k-1})}\right) }
    \\&= \ex{\log\left(\frac{(1+u_{i_k}^{k,x}) }{\sum_y q(y|i_{\le k-1})(1+ u_{i_k}^{k,y})}\right) }
    \\&\le \ex{\log(1+u_{i_k}^{k,x}) -\sum_y q(y|i_{\le k-1})\log( 1+ u_{i_k}^{k,y})}
     \\&= \ex{\log(1+u_{i_k}^{k,x})} -\sum_y \ex{q(y|i_{\le k-1})\log( 1+ u_{i_k}^{k,y})}.
\end{align*}
The first term can be upper bounded using the inequality $\log(1+x)\le x$ verified for all $x\in (-1,+\infty)$:
\begin{align*}
    \ex{\log(1+u_{i_k}^{k,x})}&= \mathds{E}_{x,i\sim q} \log(1+u_{i_k}^{k,x})
    \\&\le\mathds{E}_{x,i\sim q} u_{i_k}^{k,x}
    =\mathds{E}_{x,i\sim q_{\le k}} u_{i_k}^{k,x}
   \\& = \mathds{E}_{x,i\sim q_{\le k-1}}  \sum_{i_k}\frac{\lambda_{i_k}^k}{d} (1+u_{i_k}^{k,x})u_{i_k}^{k,x}
   \\&= \mathds{E}_{x,i\sim q_{\le k-1}}  \sum_{i_k}\frac{\lambda_{i_k}^k}{d} (u_{i_k}^{k,x})^2
\end{align*}
because $\sum_{i_k}\frac{\lambda_{i_k}^k}{d}u_{i_k}^{k,x}= \tr(d\cP_x^{m_t}(\rho_t)-\mathds{I})=0$. 
 The second term can be upper bounded using the inequality $-\log(1+x)\le \frac{1}{2}x^2-x$ verified for all $x\in (-1/2,+\infty)$:
\begin{align*}
  \ex{-\sum_y q(y|i_{\le k-1})\log( 1+ u_{i_k}^{k,y}) }
 &= -\sum_y\mathds{E}_{x,i\sim q}   q(y|i_{\le k-1})\log( 1+ u_{i_k}^{k,y})
   \\&=-\sum_y\mathds{E}_{x,i\sim q_{\le k-1}}  q(y|i_{\le k-1})\sum_{i_k}\!\frac{\lambda_{i_k}^k}{d}(1\!+\!u_{i_k}^{k,x})\log( 1\!+\! u_{i_k}^{k,y})
   \\&\le \!\sum_y\!\mathds{E}_{x,i\sim q_{\le k-1}}  \!q(y|i_{\le k-1})\!\sum_{i_k}\!\frac{\lambda_{i_k}^k}{d}(1\!+\!u_{i_k}^{k,x})(  \tfrac{1}{2}(u_{i_k}^{k,y})^2\!-\!u_{i_k}^{k,y})
    \\&\le\sum_y\mathds{E}_{x,i\sim q_{\le k-1}}  q(y|i_{\le k-1})\sum_{i_k}\frac{\lambda_{i_k}^k}{d}( (u_{i_k}^{k,x})^2+ (u_{i_k}^{k,y})^2)
    \\&=2\sum_y\mathds{E}_{x,i\sim q_{\le k-1}}  q(y|i_{\le k-1})\sum_{i_k}\frac{\lambda_{i_k}^k}{d} (u_{i_k}^{k,x})^2
    \\&=2\mathds{E}_{x,i\sim q_{\le k-1}}  \sum_{i_k}\frac{\lambda_{i_k}^k}{d} (u_{i_k}^{k,x})^2.
\end{align*}
Since the conditional mutual is upper bounded by the sum of these two terms, the upper bound on the conditional mutual information follows.
\end{proof}
The following lemma permits to conclude the upper bound on the conditional mutual information and thus the upper bound on the mutual information.
\begin{lemma}\label{lem} Let $m\ge 1$, $\cN_1,\dots, \cN_{m-1}$ be unital quantum channels and $\cP$ be a Pauli quantum channel in the family $\cF$. We have for all quantum states $\rho$ and vectors $\ket{\phi}\in\cS^{d}$:
\begin{align*}
  | \bra{\phi} d \cP \cN_{m-1}\cP \dots  \cP  \cN_{1}  \cP(\rho) \ket{\phi}-1| \le (4\eps)^m.
\end{align*}
\end{lemma}

\begin{proof}[Proof of Lemma~\ref{lem}]

For $x \in [M]$, we define the map $\cM_x$ satisfying the following equality:
\begin{align*}
    \cM_x(\rho) \coloneqq  \cP_x(\rho)-\tr(\rho)\frac{\mathds{I}}{d}=\sum_{P\in \{\mathds{I},X,Y,Z\}^{\otimes n}} \frac{4\alpha_x(P)\eps}{d^2}P\rho P ,
\end{align*}
where we have used the fact (see Lemma~\ref{int-Pauli}) that for all $\rho$:
\begin{align*}
    \sum_{P\in \{\mathds{I},X,Y,Z\}^{\otimes n}} P\rho P = d\tr(\rho)\mathds{I}.
\end{align*}
Note that $\tr(\cM_x(\rho))=\tr(\cP_x(\rho))-\tr(\rho)\tr(\frac{\mathds{I}}{d})=\tr(\rho)-\tr(\rho)=0$.
Applying a unital quantum channel $\cN$ between two quantum channels $\cP_x$ can be seen as :
\begin{align*}
    \cP_x \cN \cP_x( \rho)&= \cP_x \cN \left(\tr(\rho)\frac{\mathds{I}}{d}+ \cM_x(\rho) \right)
    \\&=\cP_x  \left(\cN\left(\tr(\rho)\frac{\mathds{I}}{d}\right)+ \cN\cM_x(\rho) \right)
    \\&=  \tr(\rho)\frac{\mathds{I}}{d}+ \cM_x\left(\tr(\rho)\frac{\mathds{I}}{d}+\cN\cM_x(\rho)\right)
    \\&=\tr(\rho)\frac{\mathds{I}}{d}+ \cM_x\cN\cM_x(\rho)
\end{align*}
because $\tr(\cN\cM_x(\rho))=\tr( \cM_x(\rho))=0 $ and 
\begin{align}\label{eq-ind}
    \cM_x(\mathds{I})&= \sum_{P\in \{\mathds{I},X,Y,Z\}^{\otimes n}} \frac{4\alpha_x(P)\eps}{d^2}\mathds{I} \notag
    \\&=\sum_{P\in \{\mathds{I},X,Y,Z\}^{\otimes n}/\sigma} \frac{4\alpha_x(P)\eps}{d^2}\mathds{I}+\frac{4\alpha_x(\sigma(P))\eps}{d^2}\mathds{I} =0.
\end{align}
By induction, we generalize the equality \eqref{eq-ind} to $m$ applications of the Pauli channel $\cP_x$:
\begin{align*}
    &\underbrace{\cP_x  \cN_{m-1}\cP_x \dots  \cP_x  \cN_{1}  \cP_x (\rho)}_{\cP_x \text{ is applied } m \text{ times}}
    =\tr(\rho)\frac{\mathds{I}}{d}+\underbrace{\cM_x  \cN_{m-1}\cM_x \dots  \cM_x \cN_{1}  \cM_x (\rho)}_{\cM_x \text{ is applied } m \text{ times}}.
\end{align*}
Therefore
\begin{align*}
    \bra{\phi} d \cP \cN_{m-1}\cP \dots  \cP  \cN_{1}  \cP(\rho) \ket{\phi} 
  &= \bra{\phi} \mathds{I}+d \cM\cN_{m-1}\cM\dots  \cM \cN_{1}  \cM(\rho) \ket{\phi}
    \\&= 1+d\bra{\phi}  \cM \cN_{m-1}\cM\dots  \cM \cN_{1}  \cM(\rho) \ket{\phi}.
\end{align*}
On the other hand, for all  vectors $\ket{\phi}\in\cS^{d}$ and Hermitian matrices $X=\sum_i \lambda_i \proj{\phi_i}$ we have: $|\bra{\phi}X\ket{\phi}|=|\sum_i \lambda_i |\spr{\phi}{\phi_i}|^2|\le\sum_i |\lambda_i| |\spr{\phi}{\phi_i}|^2=\bra{\phi}|X|\ket{\phi}   $ therefore using Lemma~\ref{int-Pauli}:
\begin{align}\label{eq-ind2}
   | \bra{\phi}  \cM(X) \ket{\phi}|  &= \left|\bra{\phi}\sum_{P\in \mathds{P}_n}  \frac{4\alpha(P)\eps}{d^2}PX P    \ket{\phi}\right|\notag
   \\& \le   \frac{4\eps}{d^2}\sum_{P\in \mathds{P}_n}|\bra{\phi}PX P \ket{\phi}| \notag
    \\&\le   \frac{4\eps}{d^2}\sum_{P\in \mathds{P}_n}\bra{\phi}P|X| P \ket{\phi}\notag
    \\&= \frac{4\eps}{d^2}\bra{\phi}  d\tr |X| \mathds{I} \ket{\phi}  
     = \frac{4\eps}{d} \tr |X|, 
\end{align}
moreover we can also obtain:
\begin{align}\label{eq-ind4}
    \tr|\cM(X)|&=\left\|\sum_{P\in \mathds{P}_n}  \frac{4\alpha(P)\eps}{d^2}PX P\right\|_1\le \sum_{P\in \mathds{P}_n}  \frac{4\eps}{d^2}\|PX P\|_1 \notag
    \\&=\sum_{P\in \mathds{P}_n}  \frac{4\eps}{d^2}\tr|X|= 4\eps \tr|X|,
\end{align}
and for a quantum channel  $\cN_j$:
\begin{align}\label{eq-ind3}
    \tr  |\cN_j(X)|  &= \|\cN_j(X)\|_1=\left\|\sum_i \lambda_i \cN_j(\proj{\phi_i})\right\|_1 \notag
   \\& \le \sum_i  \left\| \lambda_i\cN_j(\proj{\phi_i})\right\|_1 = \sum_i |\lambda_i|=\tr|X|.
\end{align}
Therefore by induction we can prove:
\begin{align}\label{0<u<1}
  |\bra{\phi} d \cP \cN_{m-1}\cP \dots  \cP  \cN_{1}  \cP(\rho) \ket{\phi}-1| \notag
    &= d|\bra{\phi}  \cM\cN_{m-1}\cM\dots  \cM \cN_{1}  \cM(\rho) \ket{\phi}|\notag
    \\&\le d \frac{4\eps}{d} \tr| \cN_{m-1}\cM\dots  \cM \cN_{1}  \cM(\rho) |\notag
    \\&\le  4\eps \tr|\cM  \cN_{m-2}\dots  \cM \cN_{1}  \cM(\rho)|\notag
    \\&\le  (4\eps)^2 \tr|\cN_{m-2} \dots  \cM  \cN_{1}  \cM(\rho)| \notag
    \\&\le \left(4\eps\right)^m
\end{align}
where the first inequality follows from \eqref{eq-ind2},  the second inequality follows from  \eqref{eq-ind3} and the third inequality follows from  \eqref{eq-ind4}. 
\end{proof}
Now we can finally  upper bound the  mutual information between $X$ and $(I_1,\dots,I_N)$:
\begin{lemma}\label{upperbound on I}
Let $\eps\le 1/4$. 
The mutual information can be upper bounded as follows:
\begin{align*}
    \cI(X: I_1,\dots,I_N) =\cO(N\eps^2).
\end{align*}
\end{lemma}
\begin{proof}[Proof of Lemma~\ref{upperbound on I}]
For all $1\le t\le N$, we remark that \[ u_{i_t}^{t,x}= \bra{\phi_{i_t}^t} d\cP_x^{m_t}(\rho_t)-\mathds{I}\ket{\phi_{i_t}^t}=\bra{\phi_{i_t}^t} d \cP_x \cN_{m_t-1}\cP_x \dots  \cP_x  \cN_{1}  \cP(\rho_t) \ket{\phi_{i_t}^t}-1,\] so by Lemma~\ref{Upper bound on cond mutual info} and Lemma~\ref{lem}:
\begin{align*}
    \cI(X:I_t| I_{\le t-1})&\le 3\mathds{E}_{x,i\sim q_{\le t-1}}  \sum_{i_t}\frac{\lambda_{i_t}^t}{d} (u_{i_t}^{t,x})^2
    \le 3\mathds{E}_{x,i\sim q_{\le k-1}}  \sum_{i_t}\frac{\lambda_{i_t}^t}{d}16\eps^2= 48\eps^2
\end{align*}
because $\sum_{i_t}\lambda_{i_t}^t=d$. Finally:
\begin{align*}
     \cI(X: I_1,\dots,I_N)=\sum_{t=1}^N \cI(X:I_t| I_{\le t-1})= \cO(N\eps^2).
\end{align*}
This concludes the proof of Lemma~\ref{upperbound on I}.
\end{proof}
Using Lemma~\ref{inequ-info-lb} and Lemma~\ref{upperbound on I} we obtain:
\begin{align*}
   \Omega(d^2) \le   \cI(X: I_1,\dots,I_N)  \le \cO(N\eps^2),
\end{align*}
which yields the lower bound $N \ge\Omega(d^2/\eps^2)$.
\end{proof}

To assess a lower bound, we need to compare it with upper bounds. The algorithm of~\cite{flammia2020efficient} implies an upper bound of $\cO\left(\frac{d^3\log(d)}{\eps^2}\right)$, 
so there is a gap between our lower bound and this upper bound. However, note that the algorithm of~\cite{flammia2020efficient} (and in fact most channel learning protocols we are aware of) use non-adaptive strategies. We will now show that indeed~\cite{flammia2020efficient} is optimal if we restrict to non-adaptive protocols.


\section{Optimal Pauli channel tomography with non-adaptive strategies}

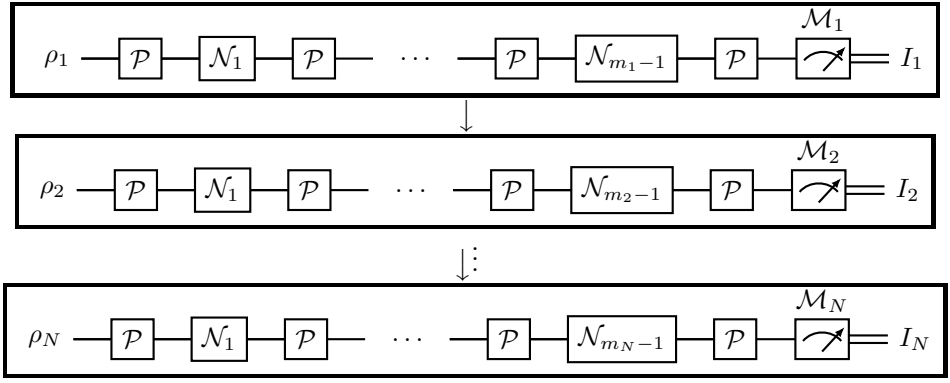
\begin{figure}
    \centering

\begin{quantikz}[thin lines] 
\gategroup[wires=1,steps=10,style={rounded corners,fill=blue!20,inner xsep=13pt},background]{}\lstick{$\rho_{1}$} & \gate[style={fill=red!40}]{\cP} \qw  & \gate[style={fill=green!50}]{\cN_1}\qw & \gate[style={fill=red!40}]{\cP} \qw &\qw~~\dots~~ & \gate[style={fill=red!40}]{\cP} \qw  &  \gate[style={fill=green!50}]{\cN_{m_1-1}} \qw& \gate[style={fill=red!40}]{\cP}\qw & \gate[style={fill=green!50}]{\cM_{1}}   & \rstick{\hspace{-0.2em}$I_1$} \cw
\end{quantikz}
\\\vspace{0.22em}
$\big\downarrow $
\\
\begin{quantikz}[thin lines] \gategroup[wires=1,steps=10,style={rounded corners,fill=blue!20,inner xsep=14pt},background]{}\lstick{$\rho_{2}$} & \gate[style={fill=red!40}]{\cP} \qw  & \gate[style={fill=green!50}]{\cN_1}\qw & \gate[style={fill=red!40}]{\cP} \qw &\qw~~\dots~~ & \gate[style={fill=red!40}]{\cP} \qw  &  \gate[style={fill=green!50}]{\cN_{m_2-1}} \qw& \gate[style={fill=red!40}]{\cP}\qw & \gate[style={fill=green!50}]{\cM_{2}}   & \rstick{\hspace{-0.2em}$I_2$} \cw
\end{quantikz} 
\\
$\big\downarrow \vdots$
\\
\begin{quantikz}[thin lines] 
\gategroup[wires=1,steps=10,style={rounded corners,fill=blue!20,inner xsep=14pt},background]{}\lstick{$\rho_{N}$\hspace{-0.2em}} & \gate[style={fill=red!40}]{\cP} \qw  & \gate[style={fill=green!50}]{\cN_1}\qw & \gate[style={fill=red!40}]{\cP} \qw &\qw~~\dots~~ & \gate[style={fill=red!40}]{\cP} \qw  &  \gate[style={fill=green!50}]{\cN_{m_{N}-1}} \qw& \gate[style={fill=red!40}]{\cP}\qw & \gate[style={fill=green!50}]{\cM_{N}}   & \rstick{\hspace{-0.2em}$I_{N}$} \cw
\end{quantikz}
    \caption{Illustration of a non-adaptive strategy for learning Pauli channel.}
    \label{fig: non-adaptive }
\end{figure}
The main difference between non-adaptive and adaptive strategies is that the former should choose the set of inputs, number of repetition, unital channels applied in between and the measurement devices before starting the learning procedure so that they cannot depend on the actual observations of the algorithm.

\begin{definition} Let $\cP$ be a Pauli channel and let $N$ be a sufficient number of steps to learn $\cP$ as defined in  \eqref{equ:pauli_channel}.  
    At step $t\in [N]$, a non-adaptive  {strategy with individual measurements} has the ability to choose an input quantum state $\rho_t$, the number $m_t\ge 1$ of uses of the quantum channel $\cP$, the unital quantum channels applied in between $\cN_1,\dots,\cN_{m_t-1}$ and the POVM $\cM_t$ for measuring the output quantum state $\rho^{\text{output}}_t$:
\begin{align*}
    \rho^{\text{output}}_t= \underbrace{\cP \circ \cN_{m_t-1}\circ\cP \circ \dots \circ \cP \circ \cN_{1} \circ \cP (\rho_t)}_{\cP \text{ is applied } m_t \text{ times}}.
\end{align*}
All these elements are chosen before starting the learning procedure (see Fig.~\ref{fig: non-adaptive } for an illustration).  By Born's rule, performing a measurement on the output quantum state $\rho^{\text{output}}_t$ using the POVM $\cM_t=\{M^t_i\}_{i\in \cI}$ is equivalent to sampling from the probability distribution 
\begin{align*}
    x_t \sim  \{\tr(\rho^{\text{output}}_tM^t_i)\}_{i\in \cI}.
\end{align*}
The observations $(x_1, \dots, x_N)$ are used to construct a probability distribution $\hat{p}$ on the set of Pauli operators $\mathds{P}_n$ satisfying with a probability at least  $2/3$:
\begin{align*}
    \TV(p, \hat{p}) \le \eps.
\end{align*}
\end{definition}
Without loss of generality, we can choose the measurement devices of   the form $\cM_t=\{\lambda_i^t \proj{\phi_i^t}\}_{i\in \cI_t}$ where $\spr{\phi_i^t}{\phi_i^t}=1  $ and $\sum_{i\in \cI_t} \lambda_i^t =d$. 
We prove the following lower bound on the total number of measurements and steps:
\begin{theorem}\label{thm:LBNA}
The problem of Pauli channel tomography using non-adaptive  ancilla-free {individual} measurements requires a total number of channel uses satisfying:
\begin{align*}
    \sum_{t=1}^N m_t \ge \Omega\left( \frac{d^4}{\eps^6}\right)
\end{align*}
or a total number of steps satisfying:
\begin{align*}
    N \ge \Omega\left( \frac{d^3}{\eps^2}\right).
\end{align*}
\end{theorem}
At a first sight we can think that this theorem is not comparable to Theorem~\ref{thm: GLB} since we give lower bounds on different parameters. However, if we ask the algorithm to only apply the channel once per step, we obtain an improved lower bound on the number of steps required for Pauli channel tomography using non-adaptive strategies. Moreover, it shows that the upper bound of \cite{flammia2020efficient}  is almost optimal especially if we know that the additional uses of channels at each step are only required to make the algorithm resilient to errors in SPAM.  
Finally, the optimal complexity $\Theta\left( \frac{d^3}{\eps^2}\right)$ for Pauli channel tomography is  surprising: We are ultimately interested in learning a classical distribution on $\mathds{P}_n \simeq[d^2]$ in $\TV$-distance which requires a complexity of $\Theta\left( \frac{d^2}{\eps^2}\right)$ in the usual sampling access model,  so our model is strictly weaker than the usual sampling access model. Furthermore, the quantum process tomography problem has an optimal copy complexity of $\tilde{\Theta}\left( \frac{d^6}{\eps^2}\right)$~\cite{oufkir-sample-optimal}: this shows that adding an additional structure to the channel can make the optimal complexity of channel tomography smaller.

\begin{proof}[Proof of Theorem~\ref{thm:LBNA}]
\label{app: proof of thm LBNA}
The construction on the family $\cF$ is similar to the construction in the proof of Theorem~\ref{thm: GLB}. We only need to add a constraint about the concentration of the mean $\frac{1}{M}\sum_{x=1}^M g(\alpha_x)$ around its expectation for a function $g$ (defined in \eqref{set G}). Let us simplify the mutual information between $X$ and $I_1,\dots,I_N$ in the non-adaptive setting. Recall from Lemma~\ref{Upper bound on cond mutual info} that the mutual information can be upper bounded as follows:
\begin{align*}
    \cI(X:I_1,\dots,I_N)&=\sum_{t=1}^N \cI(X:I_t| I_{\le t-1})
    \le 3\sum_{t=1}^N\mathds{E}_{x,i\sim q_{\le t-1}}  \sum_{i_t}\frac{\lambda_{i_t}^t}{d} (u_{i_t}^{t,x})^2.
\end{align*}
Since now we consider non-adaptive algorithms, this upper bound can be simplified:
\begin{align*}
     3\mathds{E}_{x,i\sim q_{\le t-1}}  \sum_{i_t}\frac{\lambda_{i_t}^t}{d} (u_{i_t}^{t,x})^2=3\frac{1}{M}\sum_{x=1}^M\sum_{i_t\in \cI_t}\frac{\lambda_{i_t}^t}{d} (u_{i_t}^{t,x})^2. 
\end{align*}
We remark that, in order to upper bound the mutual information $\cI(X:I_1,\dots,I_N)$, it is sufficient to  approximate 
\begin{align*}
  \frac{1}{M}\sum_{x=1}^M \sum_{t=1}^N    \sum_{i\in \cI_t}\frac{\lambda_i^t}{d}\left( \bra{\phi_{i}^t} d\cP_x^{m_t}(\rho_t)\ket{\phi_{i}^t}-1\right)^2.
\end{align*}
So the function $g$ is defined as follows
\begin{align}\label{set G}
    g(\alpha_x)= \sum_{t=1}^N    \sum_{i\in \cI_t}\frac{\lambda_i^t}{d}\left( \bra{\phi_{i}^t} d\cP_x^{m_t}(\rho_t)\ket{\phi_{i}^t}-1\right)^2
\end{align}
and we want to relate $\frac{1}{M}\sum_x g(\alpha_x)$ to $\ex{g(\alpha_x)}$.
Note that $\left( \bra{\phi} d\cP^{m_t}(\rho_t)\ket{\phi}-1\right)^2\in [0,(4\eps)^2]$ for every $\ket{\phi} \in  \cS^{d}$ and $\eps\le 1/4$ (see \eqref{0<u<1}). Also, we have for all $t\in [N]$, $\sum_{i\in \cI_t}\frac{\lambda_i^t}{d}=1$  so
\begin{align*}
    \sum_{t=1}^N    \sum_{i\in \cI_t}\frac{\lambda_i^t}{d}\left( \bra{\phi_{i}^t} d\cP_x^{m_t}(\rho_t)\ket{\phi_{i}^t}-1\right)^2\in [0, 16N\eps^2].
\end{align*}
Therefore by Hoeffding's inequality \cite{hoeff} for $s= \sqrt{\frac{(16N\eps^2)^2\log(10)}{2M}}$
\begin{align*}
    &\mathds{P} \Bigg(\bigg|\frac{1}{M}\sum_{x=1}^M\sum_{t=1}^N    \sum_{i\in \cI_t}\frac{\lambda_i^t}{d} \left( \bra{\phi^t_i} d\cP_x^{m_t}(\rho_t)\ket{\phi^t_i}-1\right)^2
 - \mathds{E}_{\alpha} \sum_{t=1}^N    \sum_{i\in \cI_t}\frac{\lambda_i^t}{d} \left( \bra{\phi_i^t} d\cP_\alpha^{m_t}(\rho_t)\ket{\phi_i^t}-1\right)^2\bigg|>s\Bigg)
\\&\le \exp\left( -\frac{2Ms^2}{(16N\eps^2)^2}\right)
    = \frac{1}{10}. 
\end{align*}
By a union bound, this error probability $1/10$ can be absorbed in the error probability of the construction by choosing a small enough constant $c$ in the cardinality of the family $M=\exp(cd^2)$.
 To recapitulate, we have proven so far that we can construct the family of quantum Pauli channels $\cF$ so that the mutual information satisfies: 
\begin{align}\label{inequ-info}
    \Omega(d^2) &\le \cI(X:I_1,\dots,I_N) \notag 
    \le 3\sum_{t}\sum_{i_t\in \cI_t} \frac{\lambda_{i_t}^t}{d}\mathds{E}_\alpha\left( \bra{\phi_{i_t}^t} d\cP_\alpha^{m_t}(\rho_t)\ket{\phi_{i_t}^t}-1\right)^2 \notag  +52N\eps^2{\exp(-cd^2)}.
\end{align}
We claim that the RHS can be upper bounded for $m_t=1$ as follows:
\begin{lemma}\label{lem:4.2}
 For all $t\in [N]$, for all  unit vectors $\ket{\phi}\in\cS^{d}$:
\begin{align*}
    \mathds{E}_\alpha\left( \bra{\phi} d\cP_\alpha(\rho_t)\ket{\phi}-1\right)^2\le  \frac{16\eps^2}{d}.
\end{align*}
\end{lemma}
If the claim is true, the inequalities~\eqref{inequ-info} imply using the fact that for all $t\le N$, $\sum_{i_t\in \cI_t} \lambda_{i_t}^t=d$:
\begin{align*}
    \Omega(d^2)&\le \cI(X:I_1,\dots,I_N)
    \le 3\sum_{t=1}^N \sum_{i_t\in \cI_t} \frac{\lambda_{i_t}^t}{d} \frac{16\eps^2}{d}+ 52N\eps^2{\exp(-cd^2)}
    \le \cO\left(N \frac{\eps^2}{d}\right)
\end{align*}
 which yields the lower bound of $N\ge \Omega(d^3/\eps^2)$  for strategies using only one channel per step.
\begin{proof}[Proof of Lemma~\ref{lem:4.2}]
Let $t\in [N]$ and $\ket{\phi}\in \cS^{d}$. We have:
\begin{align*}
    &\mathds{E}_\alpha (\bra{\phi} d\cP_\alpha(\rho_t)\ket{\phi}-1)^2
   \\& = \mathds{E}_\alpha \left( \sum_{P\in \mathds{P}_n} \frac{4\alpha(P)\eps}{d} \bra{\phi}P\rho_t P^\dagger\ket{\phi}\right)^2
    \\&= \mathds{E}_\alpha  \sum_{P,Q\in \mathds{P}_n} \frac{16\alpha(P)\alpha(Q)\eps^2}{d^2} \bra{\phi}P\rho_t P^\dagger \ket{\phi}\bra{\phi}Q\rho_t Q^\dagger \ket{\phi} 
    \\&= \sum_{P\in \mathds{P}_n} \frac{16\eps^2}{d^2} \bra{\phi}P\rho_t P^\dagger \ket{\phi}\bra{\phi}P\rho_t P^\dagger \ket{\phi} 
     -\sum_{P\in \mathds{P}_n} \frac{16\eps^2}{d^2} \bra{\phi}P\rho_t P^\dagger \ket{\phi}\bra{\phi}\sigma(P)\rho_t \sigma(P)^\dagger \ket{\phi}
    \\&\le \sum_{P\in \mathds{P}_n} \frac{16\eps^2}{d^2} \bra{\phi}P\rho_t P^\dagger \ket{\phi}^2 \le \sum_{P\in \mathds{P}_n} \frac{16\eps^2}{d^2} \bra{\phi}P\rho_t^2 P^\dagger \ket{\phi}
    =   \frac{16\eps^2}{d^2} \bra{\phi}d\tr(\rho_t^2)\mathds{I}\ket{\phi}
    \le \frac{16\eps^2}{d},
\end{align*}
where we used $\mathds{E}_\alpha \alpha(P)\alpha(Q)= 0$ if $Q\notin  \{P, \sigma(P)\}$, $\alpha(P)^2=1$, $\alpha(P)\alpha(\sigma(P))=-1$  and
the Cauchy-Schwarz inequality.
\end{proof}
Now, if we allow multiple uses of the channel at each step,
we obtain the following upper bound depending on the number $m\ge 2$ of channel uses:
\begin{lemma}\label{lem:4.3} 
For all $t\in [N]$, $m\ge 2$ and  unit vectors $\ket{\phi}\in\cS^{d}$:
\begin{align*}
    \mathds{E}_\alpha\left( \bra{\phi} d\cP_\alpha^{m}(\rho_t)\ket{\phi}-1\right)^2\le 4m\frac{(4\eps)^{2m}}{d^{\min\{2,m-1\}}}.
\end{align*}
\end{lemma}
\begin{proof}[Proof of Lemma~\ref{lem:4.3}]
Recall that for a Pauli channel $\cP_\alpha$, we can define $\cM_\alpha \coloneqq  \cP_\alpha-\tr(\cdot)\frac{\mathds{I}}{d} $ so that after $m$ applications of the Pauli channel $\cP_\alpha$ intertwined by the unital quantum channels $\cN_1,\dots,\cN_{m-1}$, we have the following identity:
\begin{align*}
    &\underbrace{\cP_\alpha  \cN_{m-1}\cP_\alpha \dots  \cP_\alpha  \cN_{1}  \cP_\alpha (\rho)}_{\cP_\alpha \text{ is applied }m \text{ times}}
  =\tr(\rho)\frac{\mathds{I}}{d}+\underbrace{\cM_\alpha  \cN_{m-1}\cM_\alpha \dots  \cM_\alpha \cN_{1}  \cM_\alpha (\rho)}_{ \cM_\alpha \text{ is applied } m \text{ times}}.
\end{align*}
The definition of $\cP_\alpha$ implies:
\begin{align*}
    \cM_\alpha(\rho)&=\cP_\alpha(\rho)-\tr(\rho)\frac{\mathds{I}}{d} 
 	= \sum_{P\in \mathds{P}_n} \frac{4\alpha(P)\eps}{d^2}P\rho P=\sum_{P\in \mathds{P}_n} \frac{4\alpha(P)\eps}{d^2}\cN_P(\rho)
\end{align*}
where we use the notation for the unital quantum channel $\cN_P(\rho)=P\rho P$ for all $P\in \mathds{P}_n$. So, using the notation $\cN_{P_m, m-1,\dots, 1,P_1}=\cN_{P_m}  \cN_{m-1}\cN_{P_{m-1}} \dots  \cN_{P_2}\cN_{1}  \cN_{P_1}$, we can develop the quantity we want to upper bound as follows:

\begin{align*}
   & \mathds{E}_\alpha\left( \bra{\phi} d\cP_\alpha^{m}(\rho)\ket{\phi}-1\right)^2
   \\&= d^2 \mathds{E}_\alpha\left( \bra{\phi}\cM_\alpha  \cN_{m-1}\cM_\alpha \dots  \cM_\alpha \cN_{1}  \cM_\alpha (\rho) \ket{\phi}\right)^2
   \\&= d^2 \mathds{E}_\alpha\Bigg(\sum_{P_1,\dots,P_m} \frac{4\alpha(P_1)\eps}{d^2}\cdots \frac{4\alpha(P_m)\eps}{d^2} 
    \cdot\bra{\phi}\cN_{P_m}  \cN_{m-1}\cN_{P_{m-1}} \dots  \cN_{P_2}\cN_{1}  \cN_{P_1} (\rho) \ket{\phi}\Bigg)^2
   \\&= \frac{(4\eps)^{2m}}{d^{4m-2}} \sum_{P,Q\in \mathds{P}_n^m} \mathds{E}_\alpha\left(\alpha(P_1)\cdots\alpha(P_m) \alpha(Q_1)\cdots\alpha(Q_m)\right) 
    \cdot\bra{\phi}\cN_{P_m, m-1,\dots, 1,P_1} (\rho) \ket{\phi}\bra{\phi}\cN_{Q_m, m-1,\dots, 1,Q_1} (\rho) \ket{\phi}.
\end{align*}
If $Q_1 \notin \left(P_1,\sigma(P_1),\dots,P_m, \sigma(P_m)\right)$ and $Q_1 \notin \left( Q_2,\sigma(Q_2),\dots, Q_m, \sigma(Q_m)\right)$ then  the  expected value 
\begin{align}
\mathds{E}_\alpha\left(\alpha(P_1)\cdots\alpha(P_m) \alpha(Q_1)\cdots\alpha(Q_m)\right)=0,
\end{align} 
as otherwise we can upper bound each term inside the sum by $1$ and we count the number of these terms. Moreover we can gain a  factor of $d^2$ by using the properties of Pauli group for $m\ge 3$. For example, suppose that $Q_1=P_1$, we have $\sum_{P\in \mathds{P}_n}\cN_P(\rho)= \sum_{P\in \mathds{P}_n} P\rho P= d\tr(\rho)\mathds{I} $ hence for $m\ge 3$, if we denote $(P,Q)_{<m}=(P_1, \dots, P_{m-1}, Q_1, \dots, Q_{m-1})$, we have:
\begin{align*}
    &\frac{(4\eps)^{2m}}{d^{4m-2}} \sum_{P,Q: Q_1=P_1} \mathds{E}_\alpha\left(\alpha(P_1)\cdots\alpha(P_m) \alpha(Q_1)\cdots\alpha(Q_m)\right)
   \cdot \bra{\phi}\cN_{P_m, m-1,\dots, 1,P_1} (\rho) \ket{\phi}\bra{\phi}\cN_{Q_m, m-1,\dots, 1,Q_1} (\rho) \ket{\phi}
    \\&\le \frac{(4\eps)^{2m}}{d^{4m-2}} \sum_{P,Q: Q_1=P_1}\bra{\phi}\cN_{P_m, m-1,\dots, 1,P_1} (\rho) \ket{\phi}
      \cdot\bra{\phi}\cN_{Q_m, m-1,\dots, 1,Q_1} (\rho) \ket{\phi}
     \\&= \frac{(4\eps)^{2m}}{d^{4m-2}} \sum_{(P,Q)_{<m}: Q_1=P_1}\bra{\phi}\sum_{P_m}\cN_{P_m, m-1,\dots, 1,P_1} (\rho) \ket{\phi}	 
     \cdot\bra{\phi}\sum_{Q_m}\cN_{Q_m, m-1,\dots, 1,Q_1} (\rho) \ket{\phi}
       \\&= \frac{(4\eps)^{2m}}{d^{4m-2}} \sum_{(P,Q)_{<m}: Q_1=P_1}\bra{\phi}d\tr(\cN_{ m-1,\dots, 1,P_1} (\rho))\mathds{I} \ket{\phi} 
      \cdot\bra{\phi}d\tr(\cN_{ m-1,\dots, 1,Q_1} (\rho))\mathds{I} \ket{\phi}
         \\& \le \frac{(4\eps)^{2m}}{d^{4m-2}} \sum_{(P,Q)_{<m}: Q_1=P_1}d^2
        = \frac{(4\eps)^{2m}}{d^{4m-2}}(d^{2})^{2m-3}d^2=\frac{(4\eps)^{2m}}{d^{2}}.
\end{align*}
Since  we have $2(2m-1)$ possibilities for $Q_1 \in \left(P_1,\sigma(P_1),\dots,P_m, \sigma(P_m), Q_2,\sigma(Q_2),\dots, Q_m, \sigma(Q_m)\right)$, we conclude that:
\begin{align*}
   & \mathds{E}_\alpha\left( \bra{\phi} d\cP_\alpha^{m}(\rho)\ket{\phi}-1\right)^2\le 2(2m-1)\frac{(4\eps)^{2m}}{d^{2}}\le 4m\frac{(4\eps)^{2m}}{d^{2}}.
\end{align*}
Now, if $m=2$, we can have $Q_1=Q_2$ and therefore we can't gain a factor $d$ when summing over $Q_2$. In this case,  
we obtain instead the upper bound:
\begin{align*}
   & \mathds{E}_\alpha\left( \bra{\phi} d\cP_\alpha^{2}(\rho)\ket{\phi}-1\right)^2\le 6\frac{(4\eps)^{4}}{d}.
\end{align*}
\end{proof}
Using the inequalities~\eqref{inequ-info} and the fact that for all $t\in [N]$, $\sum_{i_t\in \cI_t}\lambda^t_{i_t}=d$, we deduce:
\begin{align*}
    3\cdot 2^{11}\sum_{t: m_t\le 2} \frac{\eps^2}{d}+ 12\sum_{t: m_t\ge 3} m_t\frac{(4\eps)^{2m_t}}{d^2}\ge \Omega(d^2).
\end{align*}
Therefore we have either $ \sum_{t: m_t\le 2} \frac{\eps^2}{d}\ge \Omega(d^2)$ or $4\sum_{t: m_t\ge 3} m_t\frac{(4\eps)^{2m_t}}{d^2}\ge \Omega(d^2)$. Finally, we have either $N\ge \Omega\left(\frac{d^3}{\eps^2}\right)$ or $\sum_{t=1}^N m_t\ge \Omega\left(\frac{d^4}{\eps^6}\right)$.

\end{proof}

This proof relies crucially on the non-adaptiveness of the strategy. This can be seen clearly when simplifying the upper bound of the conditional mutual information in Lemma~\ref{Upper bound on cond mutual info}. For an adaptive strategy, this upper bound contains large products for which the expectation (under $\alpha$) can only upper bounded  
by $\cO(\eps^2)$ which implies a lower bound on $N$ similar to Theorem~\ref{thm: GLB}. In the next section, we explore how to overcome this difficulty in some regime of $\eps$ and improve the general lower bound $N\ge \Omega(d^2/\eps^2)$.


\section{A lower bound for Pauli channel tomography with adaptive strategies in the high precision regime}
In this section, we improve the general lower bound of quantum Pauli channel tomography in Theorem~\ref{thm: GLB} for adaptive strategies with one use of the channel each step. In the adaptive setting, a learner could adapt its choices depending on the previous observations. It can prepare a large set of inputs and measurements and thus it potentially has more power to extract information much earlier than its non-adaptive counterpart. With this intuition, we expect that lower bounds for adaptive strategies should be harder to establish. 
Since we only consider one use of the channel for each step, i.e., $m_t = 1$, a learning algorithm has the following form:

\begin{figure}
    \centering
 \begin{quantikz}[thin lines] 
					\gategroup[wires=1,steps=4,style={rounded corners,fill=blue!20,inner xsep=9pt},background]{}\fcolorbox{black}{green!50}{$\rho_{1}$} & \gate[style={fill=red!40}]{\cP} \qw& \gate[1, style={fill=green!50}]{\cM_1} \qw   &\cw\!\rstick{\hspace{-0.2em}$I_1$} 
				\end{quantikz} 
\\\vspace{0.4em}
$\big\downarrow I_1$
\\
 \begin{quantikz}[thin lines] 
					\gategroup[wires=1,steps=4,style={rounded corners,fill=blue!20,inner xsep=9pt},background]{}\fcolorbox{black}{green!50}{$\rho_{2}^{I_{1}}$} & \gate[style={fill=red!40}]{\cP} \qw& \gate[1, style={fill=green!50}]{\cM_2^{I_{1}}} \qw   &\cw\!\rstick{\hspace{-0.2em}$I_2$}
				\end{quantikz} 
\\
$\vdots \big\downarrow I_{<N}$
\\
 \begin{quantikz}[thin lines] 
					\gategroup[wires=1,steps=4,style={rounded corners,fill=blue!20,inner xsep=10pt},background]{}\fcolorbox{black}{green!50}{$\rho_{N}^{I_{<N}}$} & \gate[style={fill=red!40}]{\cP} \qw& \gate[1, style={fill=green!50}]{\cM_N^{I_{<N}}} \qw   &\cw\rstick{\hspace{-0.3em}$I_N$} 
				\end{quantikz}  
				  \caption{Illustration of an adaptive strategy for learning Pauli channel using  one channel per step.}
\label{Fig: Adap-m=1}
\end{figure}
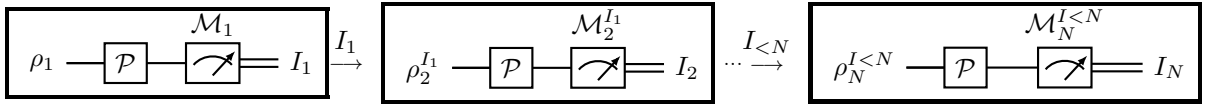

After observing $I_1,\dots,I_{t}$ at steps $1$ to $t$, the learner would choose an input $\rho_{t+1}^{\le t} \coloneqq  \rho_{t+1}^{I_1,\dots,I_t}$ and a measurement device represented by a POVM $\cM_{t+1}^{\le t}  \coloneqq \cM_{t+1}^{I_1, \dots, I_t}\coloneqq  \Big\{\lambda_{i_{t+1}}^{I_1,\dots,I_t}\proj{\phi_{i_{t+1}}^{I_1,\dots,I_t}}\Big\}_{i_{t+1}\in \cI_{t+1}^{I_1,\dots,I_t}} $ where the rank one matrices are projectors and the coefficients sum to $d$. So, the adaptive algorithm extracts classical information at step $t+1$ from the unknown Pauli quantum channel $\cP$ by first applying $\cP$ to the input $\rho_{t+1}^{I_1,\dots,I_t}$ and then performing a measurement using the POVM $\cM_{t+1}^{I_1,\dots,I_t}$ (see Fig.\ref{Fig: Adap-m=1} for an illustration). In this case, it observes $i_{t+1}\in \cI_{t+1}^{I_1,\dots,I_t}$ with a probability given by Born's rule:
\begin{align*}
   & \tr\left(\rho_{t+1}^{I_1,\dots,I_t} \lambda_{i_{t+1}}^{I_1,\dots,I_t}\proj{\phi_{i_{t+1}}^{I_1,\dots,I_t}}\right)
   = \lambda_{i_{t+1}}^{I_1,\dots,I_t}\bra{\phi_{i_{t+1}}^{I_1,\dots,I_t}}\rho_{t+1}^{I_1,\dots,I_t}\ket{\phi_{i_{t+1}}^{I_1,\dots,I_t}}.
\end{align*}
An adaptive strategy with limited {adaptivity} $N_{\text{ad}}$ can only adapt on the last previous $N_{\text{ad}}$ observations, that is for all $t\le N$: 
\begin{align*}
\rho_{t+1}^{I_1, \dots, I_t}&= \rho_{t+1}^{I_{t-N_{\text{ad}}+1}, \dots, I_t}
    \\\cM_{t+1}^{I_1, \dots, I_t}&= \cM_{t+1}^{I_{t-N_{\text{ad}}+1}, \dots, I_t}.
\end{align*}
We prove the following lower bound on the number of steps. Note that because of the assumption $m_t = 1$ for all steps $t$, the number of steps is the same as the number of channel uses.
\begin{theorem}\label{thm: LB-adaptive}
Let $\eps\le 1/(20d)$  and $d\ge 80$. Adaptive strategies for the problem of Pauli channel tomography using  one copy of the channel at each step and  ancilla-free {individual} measurements  require a number of steps $N$ satisfying:
\begin{align*}
    N\ge \Omega\left( \frac{d^{5/2}}{\eps^2}\right).
\end{align*}
Furthermore, any adaptive strategy {with limited adaptivity} $\mathcal{O}(d^2/\eps^2)$  requires a number of steps $N$ satisfying 
\begin{align}\label{equ:memory_bound}
N\ge \Omega\left( \frac{d^3}{\eps^2}\right).
\end{align}

\end{theorem}
In this theorem, we show that we can improve on the general lower bound of Theorem~\ref{thm: GLB} by an exponential factor of number of qubits if the precision  parameter $\eps$ is small enough. However, this lower bound could be as well not optimal so it remains either to improve it to match the non-adaptive upper bound of \cite{flammia2020efficient} or to propose an adaptive algorithm with a number of steps matching this lower bound. With the same proof, we can generalize this lower bound to adaptive algorithms with limited {adaptivity}. Any strategy that adapts on at most $\ceil{\frac{H}{\eps^2}}$ previous observations for the problem of Pauli channel tomography using {individual} measurements requires a number of steps $N\ge \Omega\left(\min\left\{ \frac{d^4}{\sqrt{H}\eps^2}, \frac{d^5}{H\eps^2}, \frac{d^3}{\eps^2}\right\}\right).$ For instance, if the algorithm can only adapt its input state and measurement device on the previous $\ceil{\frac{d^2}{\eps^2}}$ observations then it requires  $ N\ge \Omega\left( \frac{d^{3}}{\eps^2}\right)$ steps to correctly approximate the unknown Pauli channel. 
The remaining of this section is reserved to the proof of this theorem.


\paragraph{Construction of the family $\cF$}
We start by constructing a family of Pauli quantum channels that is $\Omega(\eps)$-separated. The elements of this family have the following form, for all $x\in \cF=[M]$:
\begin{align}\label{construction-gaussian}
    \cP_x(\rho)&= \sum_{P\in \mathds{P}_n} \frac{1+2\tilde{\alpha}_x(P)\eps d/\|\alpha_x\|_2}{d^2} P\rho P = \sum_{P\in \mathds{P}_n} p_x(P)P\rho P 
\end{align}
where $\Tilde{\alpha}_x(P)=\alpha_x(P)-\frac{1}{d^2}\sum_{Q\in \mathds{P}_n} \alpha_x(Q)$,  $\alpha_x=\left(\alpha_x(P)\right)_P$ and  $p_x(P) = \frac{1+2\tilde{\alpha}_x(P)\eps d/\|\alpha_x\|_2}{d^2}$. For $x\in \cF$, $\left(\alpha_x(P)\right)_P$ are $d^2$ random variables i.i.d. as $\cN(0,1)$. It is not difficult to check that $\{p_x\}_x$ are valid probabilities for $\eps \le 1/4d$. Indeed, for all $P\in \mathds{P}_n$ we have $|\tilde{\alpha_x}(P)|\le 2\|\alpha_x\|_2$ so for  $\eps \le 1/4d$ we have  $1+2\tilde{\alpha}_x(P)\eps d/\|\alpha_x\|_2\in [0,2]$ thus $p_x(P) \in [0, 2/d^2]\subset [0,1)$ for $d\ge 2$. 

We prove the existence of the family $\cF$ by showing that two randomly chosen Pauli channels are $\Omega(\eps)$-far with high probability.
\begin{lemma}\label{lem:p_alpha}
Let $\beta$ be a random variable independent and identically distributed as $\alpha$. We have:
\begin{align*}
    \pr{\TV(p_{\alpha},p_{\beta}) < \eps/5}\le {\exp(-cd^2)}.
\end{align*}
{for a universal constant $c>0$.}
\end{lemma}
If this claim is true, then a union bound permits to show the existence of the family with the property $M=\exp(\Omega(d^2))$. Let us  prove first a lower bound on the expected $\TV$ distance between $p_{\alpha}$ and $p_{\beta}$. 
\begin{lemma}\label{lem:exp p_alpha}
    Let $\beta$ be {a random variable independent and identically distributed as}  $\alpha$. We have:
\begin{align*}
   \ex{\TV(p_\alpha,p_\beta)}\ge  \frac{7\eps}{20}. 
\end{align*}
\end{lemma}
\begin{proof}[Proof of Lemma~\ref{lem:exp p_alpha}]
We start by writing:
\begin{align}
    \TV(p_\alpha,p_\beta)&= \frac{\eps}{d}\sum_{P\in \mathds{P}_n}\left|\frac{\tilde{\alpha}(P)}{\|\alpha\|_2}- \frac{\tilde{\beta}(P)}{\|\beta\|_2} \right| \notag
    \\&
    \ge \frac{\eps}{d}\sum_{P\in \mathds{P}_n}\left|\frac{\alpha(P)}{\|\alpha\|_2}- \frac{\beta(P)}{\|\beta\|_2} \right| \label{difference of tv1}
    \\&\quad -\frac{\eps}{d}\left|\frac{\sum_{Q\in \mathds{P}_n} \alpha(Q)}{\|\alpha\|_2}- \frac{\sum_{Q\in \mathds{P}_n} \beta(Q)}{\|\beta\|_2} \right| \label{difference of tv2}
\end{align}
where we use the triangle inequality. 

To bound the expectation  of \eqref{difference of tv2}, we use first the fact that $\|\alpha\|_2$ is independent of  $\left(\frac{\alpha(P)}{\|\alpha\|_2}\right)_P$ to show that:
\begin{align*}
    \ex{\left|\sum_{P\in \mathds{P}_n}\frac{ \alpha(P)}{\|\alpha\|_2}\right|} \ex{\|\alpha\|_2}&=\ex{\left|\sum_{P\in \mathds{P}_n}\frac{ \alpha(P)}{\|\alpha\|_2}\right|\cdot\|\alpha\|_2 }
    =\ex{ \left|\sum_{P\in \mathds{P}_n} \alpha(P)\right|}.
\end{align*}
Then, by the Cauchy-Schwarz inequality:
\begin{align*}
    \ex{\left|\sum_{P\in \mathds{P}_n} \alpha(P)\right|}\le \sqrt{\ex{\left(\sum_{P\in \mathds{P}_n} \alpha(P)\right)^2}}= d
\end{align*}
and by the Hölder's inequality:
\begin{align*}
    \ex{\|\alpha\|_2} \ge \sqrt{\frac{{\left(\ex{\|\alpha\|_2^2}\right)^3}}{\ex{\|\alpha\|_2^4}}}=\sqrt{\frac{(d^2)^3}{d^2(d^2-1)+3d^2}}\ge \frac{d}{2}.
\end{align*}
Finally, we can upper bound the expectation of \eqref{difference of tv2} as follows:
\begin{align}\label{expected tv dist1}
   & \ex{\frac{\eps}{d}\left|\frac{\sum_{Q\in \mathds{P}_n} \alpha(Q)}{\|\alpha\|_2}- \frac{\sum_{Q\in \mathds{P}_n} \beta(Q)}{\|\beta\|_2} \right| } 
   \le 2\ex{\frac{\eps}{d}\left|\frac{\sum_{P\in \mathds{P}_n} \alpha(P)}{\|\alpha\|_2}\right|} 
     =\frac{2\eps}{d} \frac{\ex{ \left|\sum_{P\in \mathds{P}_n} \alpha(P)\right|} }{\ex{\|\alpha\|_2}} \le  \frac{4\eps}{d}.
\end{align}

We move to lower bound the expectation  of \eqref{difference of tv1}.  First, using the fact that $\|\alpha\|_2$ is independent of  $\left(\frac{\alpha(P)}{\|\alpha\|_2}\right)_P$ we have: 
\begin{align*}
    \ex{\frac{\alpha(P)^2}{\|\alpha\|_2^2}}\ex{\|\alpha\|_2^2}= \ex{ \alpha(P)^2}
\end{align*}
which implies:
\begin{align}\label{first-diff-1}
    \ex{\left|\frac{\alpha(P)}{\|\alpha\|_2}- \frac{\beta(P)}{\|\beta\|_2} \right|^2 } &= 2\ex{\frac{\alpha(P)^2}{\|\alpha\|_2^2} }-2\ex{\frac{\alpha(P)\beta(P)}{\|\alpha\|_2\|\beta\|_2}}
   =2\frac{\ex{\alpha(P)^2}}{\ex{\|\alpha\|_2^2 }}=\frac{2}{d^2}.
\end{align}
Similarly, using the fact that $\|\alpha\|_2$ is independent of  $\left(\frac{\alpha(P)}{\|\alpha\|_2}\right)_P$ we have: 
\begin{align*}
    \ex{\frac{\alpha(P)^4}{\|\alpha\|_2^4}}\ex{\|\alpha\|_2^4}= \ex{ \alpha(P)^4}.
\end{align*}
This equality together with the  Hölder's inequality {(or successive Cauchy-Schwarz inequality)} imply:
\begin{align}\label{first-diff-2}
   & \ex{\left|\frac{\alpha(P)}{\|\alpha\|_2}- \frac{\beta(P)}{\|\beta\|_2} \right|^4 }\le 16\ex{\frac{\alpha(P)^4}{\|\alpha\|_2^4} }
   = 16\frac{\ex{\alpha(P)^4} }{\ex{\|\alpha\|_2^4}}= \frac{48}{d^2(d^2-1)+3d^2}\le \frac{48}{d^4}.
\end{align}
Using the inequalities \eqref{first-diff-1} and \eqref{first-diff-2} and the Hölder's inequality we obtain the following lower bound on the  expectation of \eqref{difference of tv1}:
\begin{align}\label{expected tv dist2}
    &\ex{\frac{\eps}{d}\sum_{P\in \mathds{P}_n}\left|\frac{\alpha(P)}{\|\alpha\|_2}- \frac{\beta(P)}{\|\beta\|_2} \right|}= \eps d \ex{\left|\frac{\alpha(P)}{\|\alpha\|_2}- \frac{\beta(P)}{\|\beta\|_2} \right| }\notag
    \\& \ge \eps d  \frac{\left(\ex{\left|\frac{\alpha(P)}{\|\alpha\|_2}- \frac{\beta(P)}{\|\beta\|_2} \right|^2 }\right)^{3/2} }{\left(\ex{\left|\frac{\alpha(P)}{\|\alpha\|_2}- \frac{\beta(P)}{\|\beta\|_2} \right|^4 }\right)^{1/2} }      \ge \eps d\sqrt{\frac{8/d^6}{48/d^4}}\ge \frac{2\eps}{5}. 
\end{align}
Therefore, using the inequalities \eqref{difference of tv1}, \eqref{expected tv dist1} and \eqref{expected tv dist2}, the expected value of the $\TV$-distance satisfies:
\begin{align*}
   \ex{\TV(p_\alpha,p_\beta)}\ge \frac{2\eps}{5}-\frac{4\eps}{d}\ge \frac{7\eps}{20} \quad \text{ for } d\ge 80.
\end{align*}
\end{proof}
Once we have a lower bound on the expected value of $\TV(p_\alpha,p_\beta)$, we can proceed to prove Lemma~\ref{lem:p_alpha}. 

\begin{proof}[Proof of Lemma~\ref{lem:p_alpha}]
We want to show that the function $\TV(p_\alpha,p_\beta)$ concentrates around its mean. Let ${(\gamma, \delta)}\in \left(\mathds{R}^{d^2}\right)^2$ and ${(\gamma', \delta')}\in \left(\mathds{R}^{d^2}\right)^2$ be two couples of  vectors. 
By the reverse triangle inequality we have:
\begin{align*}
    |\TV(p_\gamma,p_\delta)-\TV(p_{\gamma'},p_{\delta'})|
    &\le |\TV(p_{\gamma},p_{\delta})-\TV(p_{\gamma'},p_{\delta})|+|\TV(p_{\gamma'},p_{\delta})-\TV(p_{\gamma'},p_{\delta'})|
    \\&\le \TV(p_{\gamma},p_{\gamma'})+\TV(p_{\delta},p_{\delta'}).
\end{align*}
 Define  $E\coloneqq \{\gamma\in \mathds{R}^{d^2} : \|\gamma\|_2 > d/4\}$  
and  consider the case where  $(\gamma, \delta)\in E^2$. Recall the definition for all $P\in \mathds{P}_n$:
\begin{align*}
   p_{\gamma}(P)=  \frac{1+2\tilde{\gamma}(P)\eps d/\|\gamma\|_2}{d^2},
\end{align*}
where $\Tilde{\gamma}(P)=\gamma(P)-\frac{1}{d^2}\sum_{Q\in \mathds{P}_n} \gamma(Q).$
We have by the triangle inequality:
\begin{align*}
  \TV(p_\gamma, p_\delta) 
   &= \frac{1}{2}\sum_{P\in \mathds{P}_n} \left|\frac{1+2\tilde{\gamma}(P)\eps d/\|\gamma\|_2}{d^2}-\frac{1+2\tilde{\delta}(P)\eps d/\|\delta\|_2}{d^2}\right|
   \\&\le \frac{ \eps}{d}\sum_{P\in \mathds{P}_n} \left| \frac{\gamma(P)}{\|\gamma\|_2}-\frac{\delta(P)}{\|\delta\|_2} \right| +\frac{\eps}{d}\left|\sum_{P\in \mathds{P}_n} \frac{\gamma(P)}{\|\gamma\|_2}-\frac{\delta(P)}{\|\delta\|_2}\right|
   \\&\le \frac{ \eps}{d}\sum_{P\in \mathds{P}_n} \left| \frac{\gamma(P)}{\|\gamma\|_2}-\frac{\gamma(P)}{\|\delta\|_2} \right| +\frac{ \eps}{d}\sum_{P\in \mathds{P}_n} \left| \frac{\gamma(P)}{\|\delta\|_2}-\frac{\delta(P)}{\|\delta\|_2} \right| 
  \\ &\le\frac{ 2\eps}{d}\|\gamma\|_{1}\frac{\left| \|\gamma\|_2-\|\delta\|_2\right| }{\|\gamma\|_2\|\delta\|_2} +\frac{ 2\eps}{d}\frac{\|\gamma-\delta\|_{1}}{\|\delta\|_2} \\
   &\leq 2\eps\|\gamma\|_{2} \frac{\|\gamma-\delta\|_2}{\|\gamma\|_2\|\delta\|_2} + 2\eps \frac{\|\gamma-\delta\|_{2}}{\|\delta\|_2} 
   \\&\le 2\eps\frac{4}{d} \|\gamma-\delta\|_2 +2\eps \|\gamma-\delta\|_2 \frac{4}{d}
   = \frac{16\eps}{d} \|\gamma-\delta\|_2.
\end{align*}
Here we used that $\|\gamma\|_1\leq d\|\gamma\|_2$, as $\gamma$ is a vector with $d^2$ entries, and our assumption on the norms in the last inequality.
Hence, on the set $E^2$, by using the Cauchy-Schwarz inequality: 
\begin{align*}
    |\TV(p_\gamma,p_\delta)-\TV(p_{\gamma'},p_{\delta'})|
  &\le \TV(p_{\gamma},p_{\gamma'})+\TV(p_{\delta},p_{\delta'})
     \\&\le \frac{16\eps}{d} \|\gamma-\gamma'\|_2+\frac{16\eps}{d} \|\delta-\delta'\|_2
    \\&\le \frac{16\sqrt{2}\eps}{d} \sqrt{\|\gamma-\gamma'\|_2^2+ \|\delta-\delta'\|_2^2}
     \\&=: L \|(\gamma,\delta)-(\gamma',\delta')\|_2.
\end{align*}
Moreover, the function $(\gamma,\delta)\mapsto\TV(p_\gamma,p_\delta)$ can be extended to an $L $-Lipschitz function with respect to the $2$-norm on the whole set  $ \mathds{R}^{d^2}\times\mathds{R}^{d^2} $ using the following definition for every $(\gamma,\delta)\in \mathds{R}^{d^2}\times\mathds{R}^{d^2}$ {(Kirszbraun theorem, see App.~\ref{Kirszbraun theorem}):}
\begin{align*}
   f(\gamma,\delta)= \inf_{ (\gamma',\delta')\in E^2} \left\{ \TV(p_{\gamma'},p_{\delta'}) +L \|(\gamma,\delta)-(\gamma',\delta')\|_2\right\}.
\end{align*}
{Now consider $(\alpha, \beta)$ as a  couple of standard Gaussian vectors.}
We can control the expected value of $f(\alpha, \beta)$ using the lower bound on the expected value of  $\TV(p_\alpha,p_\beta)$ (Lemma~\ref{lem:exp p_alpha}) as follows:
\begin{align*}
    \ex{f(\alpha, \beta)}&= \ex{f{\mathbb{1}_{E^2}}(\alpha, \beta)}+\ex{f {\mathbb{1}_{(E^2)^c}}(\alpha, \beta)}
    \\&\ge \ex{f{\mathbb{1}_{E^2}}(\alpha, \beta)}
    \ge \frac{7\eps}{20}-8\eps{\exp\left(-\frac{d^2}{32} \right)}\ge \frac{3\eps}{10}
\end{align*}
because {when $(\alpha, \beta)\in \mathbb{1}_{E^2}$, we have $f(\alpha, \beta)= \TV(p_\alpha,p_\beta)$ thus}
\begin{align*}
   | \ex{f{\mathbb{1}_{E^2}}(\alpha,\beta)}-\ex{\TV(p_\alpha,p_\beta)} |
  &= {| \ex{\TV(p_\alpha,p_\beta){\mathbb{1}_{E^2}(\alpha,\beta)}}-\ex{\TV(p_\alpha,p_\beta)} |}
   \\&= \ex{\TV(p_\alpha,p_\beta){\mathbb{1}_{(E^2)^c}}(\alpha,\beta)}\le 8\eps \pr{E^c}\le 8\eps{\exp\left(-\frac{d^2}{32} \right)}
\end{align*}
where we have used the fact that $\TV(p_\alpha,p_\beta)\le 4\eps$ and  $\pr{E^c}=\pr{\|\alpha\|_2\le d/4}\le {\exp\left(-\frac{d^2}{32} \right)}$. Indeed, we can apply the concentration of Lipschitz functions of Gaussian random variables {(see App.~\ref{concentration})} for the function $\alpha\rightarrow \|\alpha\|_2$ which  is $1$-Lipschitz by the triangle inequality:
\begin{align*}
    |\|\alpha\|_2-\|\beta\|_2|\le \|\alpha-\beta\|_2
\end{align*}
and its expectation satisfies 
 $\ex{\|\alpha\|_2}\ge d/2$, thus:
\begin{align*}
    \pr{E^c}&=\pr{\|\alpha\|_2\le d/4}
    = \pr{ \|\alpha\|_2-\ex{\|\alpha\|_2}\le -d/4}
    \le {\exp\left(-\frac{d^2}{32} \right)}.
\end{align*}

We proceed with the same strategy for the function $f$ which is $L $-Lipschitz where $L=\frac{16\sqrt{2}\eps}{d}$. By the concentration of Lipschitz functions of Gaussian random variables {(see App.~\ref{concentration})}, we obtain for all $s\ge 0:$ 
\begin{align*}
    \pr{|f(\alpha,\beta)-\ex{f(\alpha,\beta)}|>s}\le {2\exp\left(-\frac{d^2s^2}{2^{10}\eps^2}\right)}
\end{align*}
Then, we can deduce the upper bound on the probability:
\begin{align*}
    \pr{\TV(p_\alpha,p_\beta) < \eps/5}
   &=\pr{\TV(p_\alpha,p_\beta) < \eps/5, (\alpha,\beta)\in E^2} +\pr{\TV(p_\alpha,p_\beta) < \eps/5, (\alpha,\beta)\notin E^2}
    \\&\le \pr{f(\alpha,\beta) < \eps/5, (\alpha,\beta)\in E^2} +\pr{(\alpha,\beta)\notin E^2}
    \\&\le \pr{f(\alpha,\beta)-\ex{f} < -\eps/10} +2\pr{\alpha\notin E}
    \\&\le{6\exp\left(-\frac{d^2}{2^{10}\cdot 10^2}\right)}\le \exp(-{cd^2})
\end{align*}
{for a universal constant $c>0$.}
\end{proof}

Hence we construct an $\eps/5$-separated family $\cF$ of {cardinality} $\exp(\Omega(d^2))$. By changing $\eps\leftrightarrow 5\eps $ in the definition of $\{\cP_x\}_{x\in \cF}$, the family becomes $\eps$-separated for $\eps\le 1/(20d)$ {and $d\ge 80$}. 

Once the family $\cF$ is constructed, we can use it to encode a message in $[M]$  to the sequence of  outcomes produced by  the learning algorithm
 when provided with the quantum Pauli channel $\cP= \cP_{x}$ in the family $\cF$. 
 More precisely, the learning algorithm  
chooses its inputs states,  performs adaptive individual measurements, and
observes a sequence of outcomes that will be transmitted to the decoder. 
Upon receiving this sequence of outcomes, the decoder runs the data-processing part of the learning algorithm to learn the Pauli channel $\cP= \cP_{x}$.  
 Therefore a $1/3$-correct algorithm can decode with a probability of failure at most $1/3$ by finding the closest quantum Pauli channel in the family $\cF$ to the channel approximated by the algorithm. By Fano's inequality, the encoder and decoder should share at least $\Omega(\log(M))\ge \Omega(d^2)$ nats of information. {More precisely, if we}   denote by $X$ the uniform random variable on the set $[M]$ representing the {message being encoded} and $I_1,\dots,I_N$ the sequence of outcomes produced by the data-acquisition part of the learning algorithm, we have: 
 \begin{lemma}\label{fano-adaptive}
The mutual information between the encoder and the {outcomes produced by the learning algorithm} is at least
 \begin{align*}
{\cI(X: I_1, \dots ,I_N)} \ge {(2/3)} \log(M) -\log(2)\ge \Omega(d^2).
\end{align*}
\end{lemma}

This lemma is similar to Lemma \ref{inequ-info-lb}. However, the constructions we use to prove these lemmas as well as the type of algorithms (or the distributions of the outcomes) are quite different.


\paragraph{Upper bound on the mutual information}
Since we have a lower bound on the mutual information, it remains to prove an upper bound depending on the number of steps $N$ and the precision $\eps$. 
By upper bounding the mutual information between $X$ and $I_1,\dots,I_N$ and using a contradiction argument, we prove Theorem~\ref{thm: LB-adaptive} which we recall:
\begin{theorem}[Restatement of Theorem~\ref{thm: LB-adaptive}]
Let  $\eps\le 1/(20d)$  and $d\ge 80$. Adaptive strategies for the problem of Pauli channel tomography using  one copy of the channel at each step and  ancilla-free {individual} measurements  require a number of steps $N$ satisfying:
\begin{align*}
    N \ge  \Omega\left( \frac{d^{5/2}}{\eps^2}\right).
\end{align*}
Furthermore, any adaptive strategy {with limited adaptivity} $\mathcal{O}(d^2/\eps^2)$  requires a number of steps $N$ satisfying 
\begin{align*}
N \ge  \Omega\left( \frac{d^3}{\eps^2}\right).
\end{align*}
\end{theorem}
\begin{proof}[Proof of Theorem~\ref{thm: LB-adaptive}] For a random vector $Y \coloneqq (Y_1, \dots, Y_M)$, we use the notation $\mathds{E}_{x}(Y)=\frac{1}{M}\sum_{x=1}^M Y_x $. For $k\in [N]$ and  a random vector $Y$ indexed by $i_1, i_2, \dots, i_k$, we use the notation $\mathds{E}_{i\sim q_{\le k-1}}(Y)= \sum_{i_1, \dots, i_{k-1}} \left(\prod_{t=1}^{k-1} \lambda_{i_t}^t\right)\cdot Y_{i_1, \dots, i_{k-1}}$.
Recall that we can write the mutual information as: $\cI(X: I_1,\dots,I_N)=\sum_{k=1}^N\cI(X:I_k| I_{\le k-1})$. 
Fix $k\in [N]$, by Lemma~\ref{Upper bound on cond mutual info}, we can upper bound  the conditional mutual information:
\begin{align}\label{eq:upper-bound-cond-mut-inf}
    \cI(X:I_k| I_{\le k-1})\le 3\mathds{E}_{x}\mathds{E}_{i\sim q_{\le k-1}} \Bigg[ \sum_{i_k}\frac{\lambda_{i_k}^k}{d} (u_{i_k}^{k,x})^2\Bigg],
\end{align}
where we use the notation
\begin{align*}
    u_{i_k}^{k,x}&=\bra{\phi_{i_k}^k} \left(d\cP_x(\rho_k)-\mathds{I}\right)\ket{\phi_{i_k}^k}
    \\&=\bra{\phi_{i_k}^k} \left(\sum_{P\in \mathds{P}_n} \frac{2\tilde{\alpha}_x(P)\eps}{\|\alpha_x\|_2} P\rho_k P\right)\ket{\phi_{i_k}^k}
    \\&=\sum_{P\in \mathds{P}_n} \frac{2\alpha_x(P)\eps}{\|\alpha_x\|_2}\bra{\phi_{i_k}^k} P\rho_kP\ket{\phi_{i_k}^k}
    -\sum_{P,Q\in \mathds{P}_n} \frac{2\alpha_x(Q)\eps}{d^2\|\alpha_x\|_2}\bra{\phi_{i_k}^k} P\rho_kP\ket{\phi_{i_k}^k}
     \\&=\sum_{P\in \mathds{P}_n} \frac{2\alpha_x(P)\eps}{\|\alpha_x\|_2}\bra{\phi_{i_k}^k} P\rho_kP\ket{\phi_{i_k}^k}-\sum_{P\in \mathds{P}_n} \frac{2\alpha_x(P)\eps}{d\|\alpha_x\|_2}.
\end{align*}
Note that for adaptive strategies the vectors  $\ket{\phi_{i_k}^k}=\ket{\phi_{i_k}^k(i_1,\dots, i_{k-1})}$ and the states  $\rho_k=\rho_k(i_1,\dots, i_{k-1})$ depend on the previous observations $(i_1,\dots, i_{k-1})$ for all $k\in [N]$. Similarly, {for a vector $\alpha= \left(\alpha(P)\right)_{P\in \mathds{P}_n}$} we denote:
\begin{align*}
    u_{i_k}^{k,\alpha}&=\sum_{P\in \mathds{P}_n} \frac{2\alpha(P)\eps}{\|\alpha\|_2}\bra{\phi_{i_k}^k} P\rho_kP\ket{\phi_{i_k}^k}-\sum_{P\in \mathds{P}_n} \frac{2\alpha(P)\eps}{d\|\alpha\|_2}
    \\&=\frac{2}{\|\alpha\|_2}\sum_{P\in \mathds{P}_n} \alpha(P)\eps\bra{\phi_{i_k}^k} P(\rho_k-\mathds{I}/d)P\ket{\phi_{i_k}^k}.
\end{align*}
 We have $\sum_{i_k} \lambda_{i_k}^k u_{i_k}^{k,x}=\tr\left(d\cP_x(\rho_k)-\mathds{I}\right)=0$ as the Pauli channel $ \cP_x$ is trace preserving.

Our goal is to bound the expectation in \eqref{eq:upper-bound-cond-mut-inf}. Since $x\sim \unif[M]$ and $(\alpha_x)_x$  are random variables  i.i.d.\ as $\alpha$, we can see the RHS of \eqref{eq:upper-bound-cond-mut-inf} as an empirical mean.
 Note that the cardinality of the constructed family $M=|\cF|$ is of order $\exp(\Omega(d^2))$, so  every {empirical mean $\mathds{E}_x g(\alpha_x)$ of a bounded function $g$ can be approximated by the expected value $\mathds{E} g(\alpha)$} with $\alpha$ following the distribution explained in the construction, the difference will be, by Hoeffding's inequality,  of order $\exp(-\Omega(d^2))$ so negligible.
 More formally, in Proposition~\ref{prop:approx}, it is shown that there is a universal constant $C>0$ such that with probability at least $9/10$ we have:
     \begin{align}\label{inequ-approx}
        & \sum_{k=1}^N\frac{1}{M}\sum_{x=1}^M\sum_{i_1,\dots,i_{k-1}} \left(\prod_{t=1}^{k-1}\lambda_{i_t }^t\left(\frac{1+u_{i_t}^{t,x}}{d}\right)\right)\sum_{i_k}\frac{\lambda_{i_k}^k}{d}(u_{i_k}^{k,x})^2 \notag
        \\&\le \sum_{k=1}^N\mathds{E}_\alpha\Bigg[\sum_{i_1,\dots,i_{k-1}} \left(\prod_{t=1}^{k-1}\lambda_{i_t }^t\left(\frac{1+u_{i_t}^{t,\alpha}}{d}\right)\right)\sum_{i_k}\frac{\lambda_{i_k}^k}{d}(u_{i_k}^{k,\alpha})^2\Bigg] 
         +N\eps^2\exp(-Cd^2).
    \end{align} 
The proof relies on Hoeffding's inequality applied on the random variable 
\begin{align*}
    \sum_{k=1}^N\sum_{i_1,\dots,i_{k-1}} \left(\prod_{t=1}^{k-1}\lambda_{i_t }^t\left(\frac{1+u_{i_t}^{t,x}}{d}\right)\right)\sum_{i_k}\frac{\lambda_{i_k}^k}{d}(u_{i_k}^{k,x})^2
\end{align*}
that is bounded by $16N\eps^2$ (as $\frac{1}{d}\sum_{i_t} \lambda_{i_{t}}^{t}(u_{i_t}^{t,\alpha})^2\le 16\eps^2$, see Lemma \ref{lem: useful}). So the random variable $$\frac{1}{M}\sum_{x=1}^M \sum_{k=1}^N\sum_{i_1,\dots,i_{k-1}} \left(\prod_{t=1}^{k-1}\lambda_{i_t }^t\left(\frac{1+u_{i_t}^{t,x}}{d}\right)\right)\sum_{i_k}\frac{\lambda_{i_k}^k}{d}(u_{i_k}^{k,x})^2$$ is essentially an empirical mean of i.i.d. bounded random variables.

Now since the inequality \eqref{inequ-approx} holds with probability at least $9/10$, we can ask in our construction \eqref{construction-gaussian} that the random vectors  $(\alpha_x)_{x\in \mathcal{F}} $ satisfy also  the inequality \eqref{inequ-approx}. The existence of such family is  guaranteed  by the union bound as the total error probability is at most $1/10+ {\exp(-cd^2)}<1$ (Lemma \ref{lem:p_alpha}).

Therefore, using the inequality \eqref{inequ-approx}, we obtain the upper bound on the mutual information:
 \begin{align}\label{equ:approximation_mutual_mcdiarmid}
   \sum_{k=1}^N \cI(X:I_k| I_{\le k-1})\notag 
   &\le 3\sum_{k=1}^N\mathds{E}_{x,i\sim q_{\le k-1}}  \sum_{i_k}\frac{\lambda_{i_k}^k}{d} (u_{i_k}^{k,x})^2\nonumber
    \\&= 3\sum_{k=1}^N \frac{1}{M}\sum_{x=1}^M\!\sum_{i_1,\dots,i_{k-1}}\! \left(\!\prod_{t=1}^{k-1}\lambda_{i_t }^t\left(\frac{1+u_{i_t}^{t,x}}{d}\right)\!\right)\!\sum_{i_k}\!\frac{\lambda_{i_k}^k}{d}(u_{i_k}^{k,x})^2\nonumber
    \\&\le 3\sum_{k=1}^N\mathds{E}_\alpha\Bigg[\!\sum_{i_1,\dots,i_{k-1}}\! \left(\prod_{t=1}^{k-1}\lambda_{i_t }^t\left(\!\frac{1+u_{i_t}^{t,\alpha}}{d}\right)\!\right)\sum_{i_k}\frac{\lambda_{i_k}^k}{d}(u_{i_k}^{k,\alpha})^2\Bigg] +3N\eps^2\exp(-Cd^2)
    \\&= 3\sum_{k=1}^N\mathds{E}_{\le k}\mathds{E}_\alpha\Bigg[ \left(\prod_{t=1}^{k-1}\left(1+u_{i_t}^{t,\alpha}\right)\right)(u_{i_k}^{k,\alpha})^2\Bigg]+ 3N\eps^2\exp(-Cd^2)\notag
\end{align}
where we use the notation ${\mathds{E}_{\le k}[Y(i_1,\dots, i_k)]  =  \frac{1}{d^k}\sum_{i_1,\dots,i_k} {\left(\prod_{t=1}^k \lambda_{i_{t}}^{t}\right)} Y(i_1,\dots, i_k)}$. 
Observe that for non-adaptive strategies, we can simplify these large products using the fact $u_{i_t}^{t,\alpha}$ does not depend on $(i_1, i_2,\dots, i_{t-1})$.
We obtain in this case an upper bound on the mutual information:
\begin{align*}
    \cI(X:I_1,\dots,I_N)\!&\le\! 3\!\sum_{k=1}^N\mathds{E}_{ k, \alpha}\!\Big[ (u_{i_k}^{k,\alpha})^2\Big]\!+\!3N\eps^2\exp(-Cd^2),
\end{align*}
{where we use the notation ${\mathds{E}_{k}[Y( i_k)]  =  \frac{1}{d}\sum_{i_k}  \lambda_{i_{k}}^{k} Y(i_k)}$.}
For this expression, using methods similar to the proof of Theorem~\ref{thm:LBNA}, one can obtain a bound of the form $\mathds{E}_{ k}\mathds{E}_\alpha\Big[ (u_{i_k}^{k,\alpha})^2\Big] = \cO(\frac{\eps^2}{d})$ which would lead to a lower bound of $\Omega(\frac{d^3}{\eps^2})$ as in Theorem~\ref{thm:LBNA}. However, for adaptive strategies, we can not simplify the terms $(1+u_{i_t}^{t,\alpha})$ for $t<k$ because $(u_{i_k}^{k,\alpha})^2$ depends on the previous observations $(i_1,\dots, i_{k-1})$. For this reason, we use Gaussian integration by parts (see Theorem~\ref{thm:GIBP}) to break the dependency between the variables in the last expectation.   Recall that for all $t,i_t$, $\Tilde{\rho_t}=\rho_t-\mathds{I}/d$ and:
\begin{align*}
    u_{i_t}^{t,\alpha}&=\frac{2}{\|\alpha\|_2}\sum_{P\in \mathds{P}_n} \alpha(P)\eps\bra{\phi_{i_t}^t} P\tilde{\rho}_tP\ket{\phi_{i_t}^t}.
\end{align*}
 Using the fact that $\|\alpha\|_2$ is independent of  $\{\alpha(P)/\|\alpha\|_2\}_P$,  {we have, for fixed $i_1, \dots, i_k$ (the expectation is on $\alpha$)}:
\begin{align*}
   &\mathds{E}_{\alpha}\left(\|\alpha\|_2^2\right) \mathds{E}_{\alpha}\left(\left(\prod_{t=1}^{k-1} \left(1+u_{i_t}^{t,\alpha}\right)\right) (u_{i_k}^{k,\alpha})^2\right)
   \\&= 2\eps \sum_{P\in \mathds{P}_n}\bra{\phi_{i_k}^k} P\tilde{\rho}_kP\ket{\phi_{i_k}^k} \mathds{E}_{\alpha}\left(\|\alpha\|_2^2\right)
    \mathds{E}_{\alpha}\left( \frac{\alpha(P)}{\|\alpha\|_2} \left(u_{i_k}^{k,\alpha}\right) \prod_{t=1}^{k-1} \left(1+u_{i_t}^{t,\alpha}\right)\right)
  \\&= 2\eps \sum_{P\in \mathds{P}_n}
  \bra{\phi_{i_k}^k} P\tilde{\rho}_kP\ket{\phi_{i_k}^k} 
    \mathds{E}_{\alpha}\left( \alpha(P) \left(\|\alpha\|_2u_{i_k}^{k,\alpha}\right) \prod_{t=1}^{k-1} \left(1+u_{i_t}^{t,\alpha}\right)\right)
     \\&=2\eps \sum_{P\in \mathds{P}_n}\bra{\phi_{i_k}^k} P\tilde{\rho}_kP\ket{\phi_{i_k}^k}\mathds{E}_{\alpha}\left( \alpha(P)F(\alpha)\right),
\end{align*}
where {$\Tilde{\rho_k}=\rho_k-\mathds{I}/d$ and }$F(\alpha)=\left(\|\alpha\|_2 u_{i_k}^{k,\alpha}\right) \prod_{t=1}^{k-1} \left(1+u_{i_t}^{t,\alpha}\right)$. {Note that the term $\bra{\phi_{i_k}^k} P\tilde{\rho}_kP\ket{\phi_{i_k}^k} $ can be factored out of the expectation because it does not depend on $\alpha$ but only on the previous observations $(i_1, \dots, i_{k-1})$ which are fixed here.}

Gaussian integration by parts (see Theorem~\ref{thm:GIBP}) implies:
 \begin{align*}
     \mathds{E}_{\alpha}\left( \alpha(P)F(\alpha)\right)&=\mathds{E}_{\alpha}\left( \partial_P F(\alpha)\right)
     \\&= 2\eps\bra{\phi_{i_k}^k} P\tilde{\rho}_kP\ket{\phi_{i_k}^k}\mathds{E}_{\alpha}\left( \prod_{t=1}^{k-1} \left(1+u_{i_t}^{t,\alpha}\right)\right)  +\sum_{s=1}^{k-1}\mathds{E}_{\alpha}\left(\|\alpha\|_2 u_{i_k}^{k,\alpha} \cdot {\left(\partial_P u_{i_s}^{s,\alpha}\right)} \cdot   \prod_{t\in  [k-1] \setminus{s}} \left(1+u_{i_t}^{t,\alpha}\right)\right).
 \end{align*}
Moreover, we have 
\begin{align*}
    \|\alpha\|_2\partial_P  u_{i_s}^{s,\alpha}
    &=2\frac{\bra{\phi_{i_s}^s} P\tilde{\rho}_sP\ket{\phi_{i_s}^s}\eps\|\alpha\|_2 }{\|\alpha\|_2}
    - 2\frac{ \partial_P\|\alpha\|_2\sum_{P\in \mathds{P}_n} \alpha(P)\bra{\phi_{i_s}^s} P\tilde{\rho}_sP\ket{\phi_{i_s}^s}\eps  }{\|\alpha\|_2}
    \\&=2\bra{\phi_{i_s}^s} P\tilde{\rho}_sP\ket{\phi_{i_s}^s}\eps-\frac{1}{\|\alpha\|_2}\alpha(P) u_{i_s}^{s,\alpha}
\end{align*}
and recall the notation $${\mathds{E}_{\le k}[Y(i_1,\dots, i_k)]  \coloneqq  \frac{1}{d^k}\sum_{i_1,\dots,i_k}{\left(\prod_{t=1}^k \lambda_{i_{t}}^{t}\right)}Y(i_1,\dots, i_k)},$$ hence
\begin{align*}
&\mathds{E}_{\le k}\mathds{E}_{\alpha}\left(\left(\prod_{t=1}^{k-1} \left(1+u_{i_t}^{t,\alpha}\right)\right)(u_{i_k}^{k,\alpha})^2\right)
  \\& =  \mathds{E}_{\le k} \frac{2\eps}{d^2}\sum_{P\in \mathds{P}_n} \bra{\phi_{i_k}^k} P\tilde{\rho}_kP\ket{\phi_{i_k}^k}\mathds{E}_{\alpha}\left( \alpha(P)F(\alpha)\right)
     \\&=  \mathds{E}_{\le k} \frac{4\eps^2}{d^2}\sum_{P\in \mathds{P}_n}\bra{\phi_{i_k}^k} P\tilde{\rho}_kP\ket{\phi_{i_k}^k}^2\mathds{E}_{\alpha}\left( \prod_{t=1}^{k-1} \left(1+u_{i_t}^{t,\alpha}\right)\right) \tag{L1}\label{equ1}
     \\&\quad +  \mathds{E}_{\le k}\frac{4\eps^2}{d^2}\sum_{P\in \mathds{P}_n}\sum_{s=1}^{k-1} \bra{\phi_{i_k}^k} P\tilde{\rho}_kP\ket{\phi_{i_k}^k}\bra{\phi_{i_s}^s} P\tilde{\rho}_sP\ket{\phi_{i_s}^s}
     \mathds{E}_{\alpha}\left(  u_{i_k}^{k,\alpha}\prod_{t\in  [k-1] \setminus{s}} \left(1+u_{i_t}^{t,\alpha}\right)\right) \tag{L2}\label{equ2}
     \\&\quad -  \mathds{E}_{\le k}\frac{2\eps}{d^2}\sum_{P\in \mathds{P}_n}\bra{\phi_{i_k}^k} P\tilde{\rho}_kP\ket{\phi_{i_k}^k}
      \sum_{s=1}^{k-1} \mathds{E}_{\alpha}\left(\frac{\alpha(P)}{\|\alpha\|_2} u_{i_k}^{k,\alpha}u_{i_s}^{s,\alpha}\prod_{t\in  [k-1] \setminus{s}} \left(1+u_{i_t}^{t,\alpha}\right)\right) \tag{L3}\label{equ3}.
 \end{align*}
 
We analyze the latter expressions line by line. Our goal is to upper bound these terms better with some expression improving the naive upper bound $\cO(\eps^2)$ on the conditional mutual information. Let us start by line~\eqref{equ1}, we have 
\begin{align*}
    \sum_{P\in \mathds{P}_n}\frac{1}{d}\sum_{i_k} \lambda_{i_{k}}^{k}\bra{\phi_{i_k}^k} P\tilde{\rho}_kP\ket{\phi_{i_k}^k}^2&\le  \sum_{P\in \mathds{P}_n}\frac{1}{d}\cdot\tr(P\tilde{\rho}_k^2 P) \sum_{P\in \mathds{P}_n}\frac{1}{d}\cdot\tr(\tilde{\rho}_k^2)\le d,
\end{align*}
so using  $\frac{1}{d}\sum_{i_t}\lambda_{i_{k}}^{k} (1+u_{i_t}^{t,\alpha})=1$ we can upper bound the line \eqref{equ1} as follows:
\begin{align*}
    \eqref{equ1}\le  \mathds{E}_{\le k-1}\frac{4\eps^2}{d} \mathds{E}_{\alpha}\left( \prod_{t=1}^{k-1} \left(1+u_{i_t}^{t,\alpha}\right)\right) = \frac{4\eps^2}{d}.
\end{align*}
This upper bound has the same order as for non-adaptive strategies. So we expect that the contribution of line \eqref{equ1} will not affect much the overall upper bound on the conditional mutual information. Next we move to line~\eqref{equ3}, first we show a useful inequality:
    \begin{lemma}\label{lem: useful}
    Let $t\in [N]$. Recall that $u_{i_t}^{t,\alpha}
    =\frac{2}{\|\alpha\|_2}\sum_{P\in \mathds{P}_n} \alpha(P)\eps\bra{\phi_{i_t}^t} P\tilde{\rho_t}P\ket{\phi_{i_t}^t}$. We have:
    \begin{align*}
         \frac{1}{d}\sum_{i_t} \lambda_{i_{t}}^{t}(u_{i_t}^{t,\alpha})^2\le 16\eps^2.
    \end{align*}
    \end{lemma}
    Observe that if we apply this upper bound directly on the expression of the  conditional mutual information (Lemma~\ref{Upper bound on cond mutual info}) we obtain an upper bound ${\cI(X:I_1, \dots, I_N)}=\cO(N\eps^2)$ which leads to a lower bound $N\ge \Omega(d^2/\eps^2)$ similar to Theorem~\ref{thm: GLB}. Still this lemma will be useful for controlling intermediate expressions appearing for the upper bound of line~\eqref{equ3}. 
    \begin{proof} We use the fact that every {matrix $A$  can be written as $A= \sum_{R\in\mathds{P}_n}\frac{\tr(AR)}{d} R$} 
        \begin{align*}
    \sum_{i_t} \lambda_{i_{t}}^{t}(u_{i_t}^{t,\alpha})^2
   &= \frac{4\eps^2}{\|\alpha\|_2^2}\sum_{i_t} \lambda_{i_{t}}^{t}\bra{\phi_{i_t}^t}\left(\sum_{P\in \mathds{P}_n} \alpha(P) P\tilde{\rho_t}P\right)\ket{\phi_{i_t}^t}^2
   \\& \le \frac{4\eps^2}{\|\alpha\|_2^2}\tr\left( \sum_{P\in \mathds{P}_n} \alpha(P) P\tilde{\rho_t}P\right)^2
    \\&=\frac{4\eps^2}{\|\alpha\|_2^2} \tr\left( \sum_{P\in \mathds{P}_n} \alpha(P) \frac{1}{d}\sum_{R\in \mathds{P}_n} \tr(R\tilde{\rho_t}) PRP\right)^2
    \\& =\frac{4\eps^2}{\|\alpha\|_2^2} \tr\left( \sum_{P\in \mathds{P}_n} \alpha(P) \frac{1}{d}\sum_{R\in \mathds{P}_n} \tr(R\tilde{\rho_t}) (-1)^{R{\circ} P}R\right)^2
\end{align*}
where we used that $PRP = (-1)^{R{\circ} P}R$. Now we expand the square,  use $\tr(R_{1}R_{2}) = d\cdot \mathbb{1}_{R_{1}= R_{2}}$ and   Lemma~\ref{sum-Pauli} to obtain 

        \begin{align*}
    \sum_{i_t} \lambda_{i_{t}}^{t}(u_{i_t}^{t,\alpha})^2
     &=\frac{4\eps^2}{\|\alpha\|_2^2} \!\sum_{P,P',R\in \mathds{P}_n}\!\alpha(P)\alpha(P') \frac{1}{d}\cdot\tr(R\tilde{\rho_t})^2 (-1)^{R{\circ} P}(-1)^{R{\circ} P'}
     \\&=\frac{4\eps^2}{\|\alpha\|_2^2} \sum_{R\in \mathds{P}_n} \left( \sum_{P\in \mathds{P}_n} \alpha(P)(-1)^{R{\circ} P}\right)^2 \frac{1}{d}\cdot\tr(R\tilde{\rho_t})^2
     \\&\le \frac{16\eps^2}{\|\alpha\|_2^2} \sum_{P,P',R\in \mathds{P}_n} \alpha(P)\alpha(P') \frac{1}{d} (-1)^{R\circ(PP')}
     \\&=\frac{16\eps^2}{\|\alpha\|_2^2} \sum_{P,P'\in \mathds{P}_n} \alpha(P)\alpha(P') \cdot d\cdot\mathbb{1}_{PP'=\mathds{I}} 
    \\&=\frac{16d\eps^2}{\|\alpha\|_2^2} \sum_{P=P'\in \mathds{P}_n} \alpha(P)^2 =16d\eps^2.
\end{align*}
In the previous inequality, we used that for all $R\in \mathds{P}_n$ we have $\|R\|_\infty=1$ so using Hölder's inequality we deduce $|\tr(R\tilde{\rho})|\le \|R\|_\infty \|\tilde{\rho}\|_1\le 2$.
    \end{proof}
Observe that the condition $\eps\le 1/(4d)$ implies that for all $t\in[N]$ and $i_t\in \cI_t$ we have $1+u_{i_t}^{t,\alpha}\ge 1/16$ and recall that $\mathds{E}_{k}(1+u_{i_t}^{t,\alpha})=\frac{1}{d} \sum_{i_t} \lambda^t_{i_t}(1+u_{i_t}^{t,\alpha})=1 $.  
Therefore, \eqref{equ3} can be controlled  as follows:
\begin{align*}
    \eqref{equ3}
    &=- \mathds{E}_{\le k}\frac{1}{d^2}\sum_{s=1}^{k-1} \mathds{E}_{\alpha}\left[ \left(u_{i_k}^{k,\alpha}\right)^2u_{i_s}^{s,\alpha}\prod_{t\in  [k-1] \setminus{s}} \left(1+u_{i_t}^{t,\alpha}\right)\right]
      \\&= \frac{1}{d^2}  \mathds{E}_{\le k}\!\left[\!- \mathds{E}_{\alpha}\left( \!\left(\!\sum_{s<k}\! {\frac{u_{i_s}^{s,\alpha}}{1+u_{i_s}^{s,\alpha}}(u_{i_k}^{k,\alpha})^2}\right)\!\prod_{{t\in  [k-1]}}\!\left(1+u_{i_t}^{t,\alpha}\right)\!\right)\!\right]
       \\&\le \frac{1}{d^2}  \mathds{E}_{\le k}\!\left[ \mathds{E}_{\alpha}\!\left(\!   \left(\left|\sum_{s<k} {\frac{u_{i_s}^{s,\alpha}}{1+u_{i_s}^{s,\alpha}}}\right|(u_{i_k}^{k,\alpha})^2\!\right)\!\prod_{{t\in  [k-1]}}\!\left(1+u_{i_t}^{t,\alpha}\!\right)\!\right)\!\right]
        \\&\overset{(a)}{\le}  \frac{16\eps^2}{d^2}  \mathds{E}_{< k}\left[ \mathds{E}_{\alpha}\left(  \left(\left|\sum_{s<k} \frac{u_{i_s}^{s,\alpha}}{(1+u_{i_s}^{s,\alpha}) }\right|\right)\prod_{t<k} \left(1+u_{i_t}^{t,\alpha}\right)\right) \right] 
        \\& \overset{(b)}{\le}  \frac{16\eps^2}{d^2} \sqrt{ \mathds{E}_{\alpha}  \mathds{E}_{< k}   \left|\sum_{s<k} \frac{u_{i_s}^{s,\alpha}}{(1+u_{i_s}^{s,\alpha}) }\right|^2 \prod_{t<k}\left(1+u_{i_t}^{t,\alpha} \right)}
          \cdot \sqrt{\mathds{E}_{\alpha}  \mathds{E}_{< k} \prod_{t<k}\left(1+u_{i_t}^{t,\alpha} \right)} 
        \\&\overset{(c)}{=}  \frac{16\eps^2}{d^2} \sqrt{ \mathds{E}_{\alpha}  \mathds{E}_{< k}  \!\sum_{s,r<k}\! \frac{u_{i_s}^{s,\alpha}}{(1+u_{i_s}^{s,\alpha} )}\cdot \frac{u_{i_r}^{r,\alpha}}{(1+u_{i_r}^{r,\alpha}) }\!\prod_{t<k}\!\left(1+u_{i_t}^{t,\alpha} \right)} 
        \\&\overset{(d)}{=}  \frac{16\eps^2}{d^2} \sqrt{ \mathds{E}_{\alpha}  \mathds{E}_{< k}   \sum_{s<k} \frac{(u_{i_s}^{s,\alpha})^2}{(1+u_{i_s}^{s,\alpha} )^2}\prod_{t<k}\left(1+u_{i_t}^{t,\alpha} \right)} 
         \\&\overset{(e)}{\le}  \frac{64\eps^2}{d^2} \sqrt{ \mathds{E}_{\alpha}  \mathds{E}_{< k}   \sum_{s<k} (u_{i_s}^{s,\alpha} )^2 \prod_{t\in  [k-1] \setminus{s}}\left(1+u_{i_t}^{t,\alpha} \right)} 
         \\&\overset{(f)}{\le}  \sqrt{k}\frac{256\eps^3}{d^2}, 
\end{align*}
where in $(a)$ we used Lemma~\ref{lem: useful}; in $(b)$ we used Cauchy-Schwarz inequality; in $(c)$ we used $ \mathds{E}_{< k} \prod_{t<k}\left(1+u_{i_t}^{t,\alpha} \right)=1$; in $(d)$ we used $\mathds{E}_{\le\max\{s, r\}}\left(u_{i_s}^{s,\alpha}u_{i_r}^{r,\alpha}\right)=0 $ when   $s\neq r$; in $(e)$ we used  $1+u_{i_s}^{s,\alpha} \ge 1/16$ since $\eps \le 1/4d$; in $(f)$ 
we used also Lemma~\ref{lem: useful}. Indeed, we can simplify the expectation as follows 
\begin{align*}
    \mathds{E}_{< k}   \sum_{s<k} (u_{i_s}^{s,\alpha} )^2 \prod_{t\in  [k-1] \setminus{s}}\left(1+u_{i_t}^{t,\alpha} \right)
  &= \sum_{s<k} \mathds{E}_{<k}   (u_{i_s}^{s,\alpha} )^2 \prod_{t\in  [k-1] \setminus{s}}\left(1+u_{i_t}^{t,\alpha} \right)
    \\&= \sum_{s<k} \mathds{E}_{\le s}   (u_{i_s}^{s,\alpha} )^2 \prod_{t\in  [s-1] }\left(1+u_{i_t}^{t,\alpha} \right)
    \\&\le \sum_{s<k} \mathds{E}_{\le s-1}   16\eps^2 \prod_{t\in  [s-1] }\left(1+u_{i_t}^{t,\alpha} \right) && \text{(Lemma~\ref{lem: useful})}
    \\&= \sum_{s<k}    16\eps^2 \le 16\eps^2k.
\end{align*}
Finally, we control the line~\eqref{equ2} which is more involved. Let us adopt the notation for $s,k \in [N]$:
\begin{align*}
    M_{s,k}&=\sum_{P\in \mathds{P}_n}  P\tilde{\rho}_kP\bra{\phi_{i_s}^s} P\tilde{\rho}_sP\ket{\phi_{i_s}^s}
    \\&=\frac{1}{d^2}\sum_{P,Q,R\in \mathds{P}_n} \tr(\tilde{\rho}_kQ)\tr(\tilde{\rho}_s R) PQP\bra{\phi_{i_s}^s} PRP\ket{\phi_{i_s}^s} 
      \\&=\frac{1}{d^2}\!\sum_{P,Q,R\in \mathds{P}_n}\! \tr(\tilde{\rho}_kQ)\tr(\tilde{\rho}_s R)(-1)^{P{\circ}(QR)} Q\bra{\phi_{i_s}^s}\!R\!\ket{\phi_{i_s}^s} 
      \\&= \sum_{Q\in \mathds{P}_n} \tr(\tilde{\rho}_kQ)\tr(\tilde{\rho}_sQ) \bra{\phi_{i_s}^s} Q\ket{\phi_{i_s}^s}  Q,
\end{align*}
where we used Lemma~\ref{sum-Pauli} in the last equality. 
So  we can write $\sum_{P\in \mathds{P}_n} \bra{\phi_{i_k}^k} P\tilde{\rho}_kP\ket{\phi_{i_k}^k}\bra{\phi_{i_s}^s} P\tilde{\rho}_sP\ket{\phi_{i_s}^s} =\bra{\phi_{i_k}^k}M_{s,k} \ket{\phi_{i_k}^k}$.
Also we use the notation  $\Psi_k=\mathds{E}_{\le k}{\mathds{E}_{\alpha}} \left(
     \left(u_{i_k}^{k,\alpha}\right)^2 \prod_{t<k} \left(1+u_{i_t}^{t,\alpha}\right) \right)$ and $\Phi_{k}^{\alpha} =\prod_{t<k} \left(1+u_{i_t}^{t,\alpha}\right) $
so that we have the (in)equalities using Cauchy-Schwarz inequality:
\begin{align*}
    \eqref{equ2}
  &=  \frac{4\eps^2}{d^2} \mathds{E}_{\le k}\mathds{E}_{\alpha}\!\left[\sum_{s=1}^{k-1}  \bra{\phi_{i_k}^k}\! M_{s,k}\!  \ket{\phi_{i_k}^k} u_{i_k}^{k,\alpha}\!\prod_{t\in  [k-1] \setminus{s}}\! \left(1\! +\! u_{i_t}^{t,\alpha}\right)\right]
    \\&=  \frac{4\eps^2}{d^2} \mathds{E}_{\le k}\mathds{E}_{\alpha}\!\left[ \sum_{s=1}^{k-1}\frac{\bra{\phi_{i_k}^k}\! M_{s,k} \! \ket{\phi_{i_k}^k}}{(1+u_{i_s}^{s,\alpha})}  \cdot u_{i_k}^{k,\alpha}\!\prod_{t\le k-1}\!\left(1+u_{i_t}^{t,\alpha}\right)\right]
    \\&\le  \frac{4\eps^2}{d^2} \sqrt{\mathds{E}_{\le k}\mathds{E}_{\alpha}\left[ \left(\sum_{s=1}^{k-1}\frac{\bra{\phi_{i_k}^k}M_{s,k} \ket{\phi_{i_k}^k}}{(1+u_{i_s}^{s,\alpha})} \right)^2 \cdot \Phi_{k}^{\alpha}\right]\Psi_k} 
        \\&\le  \frac{4\eps^2}{d^2} \sqrt{\! \mathds{E}_{\le k}\mathds{E}_{\alpha}\left[ \bra{\phi_{i_k}^k}\! \left(\sum_{s=1}^{k-1}\frac{M_{s,k}}{(1+u_{i_s}^{s,\alpha})}\!  \right)^2\! \ket{\phi_{i_k}^k}\! \cdot \Phi_{k}^{\alpha}\right]\Psi_k}
            \\&= \frac{4\eps^2}{d^2} \sqrt{\frac{1}{d}\mathds{E}_{\le k-1}\mathds{E}_{\alpha}\! \left[ \tr\left(\! \left(\sum_{s=1}^{k-1}\frac{M_{s,k}}{(1+u_{i_s}^{s,\alpha})} \right)^2\right)\!  \cdot \Phi_{k}^{\alpha}\right]\! \Psi_k}.
\end{align*}
From the definition of $M_{s,k}$ we can write 
that for $s,t < k$
\begin{align*}
  &\tr(M_{s,k}M_{t,k}) 
  = d \sum_{Q\in \mathds{P}_n}  \tr(\tilde{\rho}_kQ)^2\tr(\tilde{\rho}_sQ)\tr(\tilde{\rho}_tQ) \bra{\phi_{i_s}^s}Q\ket{\phi_{i_s}^s} \bra{\phi_{i_t}^t} Q\ket{\phi_{i_t}^t}. 
\end{align*}
Hence
\begin{align*}
  \tr\left(\sum_{s=1}^{k-1}\frac{M_{s,k}}{(1+u_{i_s}^{s,\alpha})} \right)^2 
   & =\!\sum_{\substack{s,t < k\\ Q\in \mathds{P}_n}} \!d\tr(\tilde{\rho}_kQ)^2 \frac{\tr(\tilde{\rho}_sQ)\tr(\tilde{\rho}_tQ) \bra{\phi_{i_s}^s} \!Q\!\ket{\phi_{i_s}^s} \bra{\phi_{i_t}^t}\!Q\!\ket{\phi_{i_t}^t} }{(1+u_{i_s}^{s,\alpha}) (1+u_{i_t}^{t,\alpha})}
    \\&=d\sum_{Q\in \mathds{P}_n} \tr(\tilde{\rho}_kQ)^2 \left( \sum_{s< k} \frac{\tr(\tilde{\rho}_sQ) \bra{\phi_{i_s}^s} Q\ket{\phi_{i_s}^s}  }{(1+u_{i_s}^{s,\alpha})} \right)^2
    \\&\le 4d\sum_{Q\in \mathds{P}_n}  \left( \sum_{s< k} \frac{\tr(\tilde{\rho}_sQ) \bra{\phi_{i_s}^s} Q\ket{\phi_{i_s}^s}  }{(1+u_{i_s}^{s,\alpha})} \right)^2.
\end{align*}
Note that this step is crucial because $\rho_k$ depends on $(i_1,\dots,i_{k-1})$ so we need to avoid it in order to simplify with the expectations $\mathds{E}_{t}$ for $t < k$. When we want to simplify the expectation 
\begin{align*}
    &\mathds{E}_{\le k-1}{\mathds{E}_{\alpha}}\left( \tr\left(\left(\sum_{s=1}^{k-1}\frac{M_{s,k}}{(1+u_{i_s}^{s,\alpha})} \right)^2\right) \prod_{t\le k-1} \left(1+u_{i_t}^{t,\alpha}\right)\right) 
\\&\le  4d\sum_{Q\in \mathds{P}_n}   \mathds{E}_{\le k-1}\mathds{E}_{\alpha}\Bigg(\sum_{s_1=1}^{k-1}\sum_{s_2=1}^{k-1}\bigg(\prod_{t\le k-1} \left(1+u_{i_t}^{t,\alpha}\right) \bigg)
 \cdot \frac{\tr(\tilde{\rho}_{s_1}Q) \bra{\phi_{i_{s_1}}^{s_1}} Q\ket{\phi_{i_{s_1}}^{s_1}}  }{(1+u_{i_{s_1}}^{s_1,\alpha})}
 \cdot \frac{\tr(\tilde{\rho}_{s_2}Q) \bra{\phi_{i_{s_2}}^{s_2}} Q\ket{\phi_{i_{s_2}}^{s_2}}  }{(1+u_{i_{s_2}}^{s_2,\alpha})}  \Bigg),
\end{align*}
 we can see that if $s_1< s_2$ (or $s_1>s_2$), we will get $0$  because we can simplify the terms $ (1+u_{i_t}^{t,\alpha})$ in the product for $t>s_2$, the term  $(1+u_{i_{s_2}}^{s_2,\alpha})$ is simplified with the denominator so we can take safely the expectation under $\mathds{E}_{s_2}$:
\begin{align*}
    \mathds{E}_{s_2} \tr(\tilde{\rho}_{s_2}Q) \bra{\phi_{i_{s_2}}^{s_2}} Q\ket{\phi_{i_{s_2}}^{s_2}}
    =\frac{1}{d}\sum_{i_{s_2}}\tr(\tilde{\rho}_{s_2}Q) \lambda_{i_{s_2}}^{s_2}\bra{\phi_{i_{s_2}}^{s_2}} Q\ket{\phi_{i_{s_2}}^{s_2}}= \tr(\tilde{\rho}_{s_2}Q) \tr(Q)=0
\end{align*}
because $\tr(Q)=0$ unless $Q=\mathds{I}$ for which $\tr(\tilde{\rho}_{s_2}Q)=\tr(\tilde{\rho}_{s_2})=\tr(\rho_{s_2}-\mathds{I}/d)=0$.
Therefore
\begin{align*}
    \eqref{equ2}
   & \le  \frac{4\eps^2}{d^2} \sqrt{\frac{1}{d}\mathds{E}_{\le k-1}\mathds{E}_{\alpha}\left[ \tr\left(\sum_{s=1}^{k-1}\frac{M_{s,k}}{(1+u_{i_s}^{s,\alpha})} \right)^2 \cdot \Phi_{k}^{\alpha}\right]\Psi_k}
    \\&\le\! \frac{8\eps^2}{d^2} \!\sqrt{\!\mathds{E}_{< k}\mathds{E}_{\alpha}\!\left[\! \sum_{Q\in \mathds{P}_n} \! \left( \!\sum_{s< k} \frac{\tr(\tilde{\rho}_sQ)\! \bra{\phi_{i_s}^s} \!Q\!\ket{\phi_{i_s}^s}  }{(1+u_{i_s}^{s,\alpha})} \!\right)^2  \!\Phi_{k}^{\alpha}\!\right]\!\Psi_k}
    \\& \overset{(a)}{\le} \!\frac{32\eps^2}{d^2}\!\sqrt{ \sum_{s< k} \sum_{Q\in \mathds{P}_n}   \mathds{E}_{\le s}\tr(\tilde{\rho}_sQ)^2 \bra{\phi_{i_s}^s} \!Q\!\ket{\phi_{i_s}^s}^2  \mathds{E}_{\alpha}\left( \Phi_{s}^{\alpha}\right)\!\Psi_k} 
       \\& \overset{(b)}{\le}\!\frac{{64}\eps^2}{d^2} \sqrt{\!\sum_{s< k}\!   \mathds{E}_{\le s}\! \sum_{Q\in \mathds{P}_n} \!\bra{\phi_{i_s}^s} \!Q\!\proj{\phi_{i_s}^s} \! Q\!\ket{\phi_{i_s}^s} \!\mathds{E}_{\alpha}\!\left(\Phi_{s}^{\alpha}\right)\!\Psi_k}
        \\& \overset{(c)}{=}\!\frac{{64}\eps^2}{d^2} \sqrt{ \sum_{s< k}   \mathds{E}_{\le s}  \bra{\phi_{i_s}^s} d\tr(\proj{\phi_{i_s}^s}  ) \mathds{I} \ket{\phi_{i_s}^s} \mathds{E}_{\alpha}\left(  \Phi_{s}^{\alpha}\right)\!\Psi_k} 
         \\&\le \frac{{64}\eps^2}{d\sqrt{d}} \sqrt{ k}\sqrt{\Psi_k},
\end{align*}
where in $(a)$ we used $(1+u_{i_s}^{s,\alpha})\ge \frac{1}{16}$; in $(b)$ we used $|\tr(\tilde{\rho}_sQ)|\le 2$; in $(c)$ we used Lemma~\ref{int-Pauli}.

We have proven so far, for all $k\le N:$
\begin{align}\label{proppsi}
     \Psi_k \le \frac{4\eps^2}{d}+ 256\sqrt{k}\frac{\eps^3}{d^2}+\frac{{64}\eps^2}{d\sqrt{d}} \sqrt{ k}\sqrt{\Psi_k}. 
\end{align}
The first term of the upper bound can be seen as a non-adaptive contribution. The second one can be thought as a geometric mean of the first and third terms. The last term represents essentially the contribution of the adaptivity. Our final stage of the proof is to use these recurrence inequalities to prove the lower bound by a contradiction argument.

Recall that $\Psi_k=\mathds{E}_{\le k}\mathds{E}_{\alpha}\left[
     \left(u_{i_k}^{k,\alpha}\right)^2 \prod_{t<k} \left(1+u_{i_t}^{t,\alpha}\right)\right] $ and $\sum_{k=1}^N\cI(X:I_k|I_{<k})\le 3\sum_{k=1}^N\Psi_k+3N\eps^2\exp(-Cd^2)$.
We suppose that $N\le c\frac{d^{5/2}}{\eps^2}$ for sufficiently small $c>0$. We know that from Lemma~\ref{fano-adaptive} and Lemma~\ref{Upper bound on cond mutual info}
\begin{align*}
    c_0d^2\le \cI(X:Y)\le 3\sum_{k\le N} \Psi_k +3N\eps^2\exp(-Cd^2). 
    \end{align*}
So $\sum_{k\le N} \Psi_k\ge c' d^2$ (for example $c'=c_0/4$), on the other hand the inequality (\ref{proppsi}) implies:
\begin{align*}
    \sum_k \Psi_k &\le {\sum_k \left( \frac{4\eps^2}{d} +256\frac{\eps^3}{d^2}\sqrt{k} +
    {64}\frac{\eps^2\sqrt{k}}{d\sqrt{d}}\sqrt{\Psi_k} \right)}
    \\&\le 4\frac{N\eps^2}{d} +256\frac{N\eps^3}{d^2}\sqrt{N} +{64}\sum_k\frac{\eps^2\sqrt{k}}{d\sqrt{d}}\sqrt{\Psi_k}
    \\&\le 4\frac{N\eps^2}{\sqrt{c'}d^2} \sqrt{\sum_{k\le N} \Psi_k} +256\frac{N\eps^2}{d^2}\sqrt{cd^{5/2}}  +{64}\frac{\eps^2}{d\sqrt{d}}\sqrt{\sum_k k}\sqrt{\sum_k \Psi_k} 
    \\&\le \left( \frac{8}{\sqrt{c'd}}+\frac{512}{\sqrt{c'}d^{1/4}} +{64}\right)\frac{\eps^2}{d\sqrt{d}}\sqrt{\sum_k k}\sqrt{\sum_k \Psi_k}
    \\&\le C'\frac{\eps^2}{d\sqrt{d}}N\sqrt{\sum_k \Psi_k}
\end{align*}
where in the third inequality we use $\sum_{k\le N} \Psi_k\ge c' d^2$, $N\le c\frac{d^{5/2}}{\eps^2}$ and Cauchy-Schwarz inequality and in the last inequality $C'$ is a universal constant. Therefore:
\begin{align*}
    \sum_k \Psi_k \le C'^2\left(  \frac{N^2\eps^4}{d^3}\right)\le C'^2c^2d^2.
\end{align*}
Hence
\begin{align*}
    c_0d^2\le \cI(X:Y)\le \sum_{k\le N} 3\Psi_k +3N\eps^2\exp(-Cd^2)\le 6 C'^2 c^2d^2
\end{align*}
which gives the contradiction for $c\ll \sqrt{c_0}/C'$. Finally we deduce $N\ge \Omega(d^{5/2}/\eps^2)$ and we conclude the proof of the first lower bound of Theorem~\ref{thm: LB-adaptive}.

If the adaptive algorithm {can only adapt on the last $\cO(H/\eps^2)$ observations}, the previous inequalities imply for all $1\le k \le N$:
\begin{align*}
   \sum_k \Psi_k &\le { \sum_k \left( \frac{4\eps^2}{d}+ 256\sqrt{H}\frac{\eps^2}{d^2}+ \frac{{64}\eps}{d\sqrt{d}}\sqrt{H}\sqrt{\Psi_k} \right)}
    \\&\le  { \sum_k \left(  \frac{4\eps^2}{d}+ 256\sqrt{H}\frac{\eps^2}{d^2}+ \frac{({64}\eps)^2H}{d^3}+\frac{\Psi_k}{2}\right)}
\end{align*}
where we use AM-GM inequality, hence we deduce:
\begin{align*}
    c_0d^2&\le \cI(X:Y)\le  \sum_k 3\Psi_k +3N\eps^2\exp(-Cd^2)
    \\&\le \frac{30\eps^2N}{d}+ 6\cdot 256\sqrt{H}\frac{\eps^2N}{d^2}+ \frac{6\cdot({64}\eps)^2HN}{d^3}
\end{align*}
and finally we obtain:
\begin{align*}
    N\ge \Omega\left(\min\left\{ \frac{d^4}{\sqrt{H}\eps^2}, \frac{d^5}{H\eps^2}, \frac{d^3}{\eps^2}\right\}\right).
\end{align*}
For $H=\mathcal{O}(d^2)$, this gives \eqref{equ:memory_bound} and we conclude the proof of the second lower bound of Theorem~\ref{thm: LB-adaptive}.
\end{proof}

\section{Conclusion and open problems}

We have provided lower bounds for Pauli channel tomography in the diamond norm using  ancilla-free independent strategies for both adaptive and non-adaptive strategies. In particular, we have shown that the number of measurements should be at least $\Omega(d^3/\eps^2)$ in the non-adaptive setting and $\Omega(d^{2.5}/\eps^2)$ in the adaptive setting. We would like to finish with three interesting directions. 
Finding the optimal complexity of Pauli channel tomography using adaptive individual measurements remains an open question. We conjecture this complexity to be $\Theta(d^3/\eps^2)$ since we remark that in many situations the adaptive strategies cannot overcome the non-adaptive ones. Furthermore, we already obtained a $\Theta(d^3/\eps^2)$ bound for adaptive strategies in the high precision and  limited adaptivity regime, further evidence of this bound.
Moreover, since~\cite{chen2022quantum} established the optimal complexity for estimating the eigenvalues of a Pauli channel in the $l_\infty$-norm using ancilla-assisted non-adaptive independent strategies, it would be interesting to find the optimal complexity to learn a Pauli channel in the diamond norm when the algorithm can use $k$-qubit ancilla for $k\le n$. Finally, it should be noted that all of the channel constructions used in this work have a very large spectral gap, i.e., are very noisy. It would be interesting to study the sample complexity of Pauli channel tomography in terms of the spectral gap as well.

\section*{Acknowledgment}
A.O. thanks Guillaume Aubrun for helpful discussions. We thank the anonymous reviewers for their thorough comments, which significantly improve the presentation.
This work is part of HQI initiative (\href{www.hqi.fr}{www.hqi.fr}) and is supported by France 2030 under the French National Research Agency award number “ANR-22-PNCQ-0002”. We also acknowledge support from the European Research Council (ERC Grant AlgoQIP, Agreement No. 851716).

\printbibliography
\appendix
\section{Technical tools}
\subsection{Pauli group properties}
In this section, we group some useful properties about the Pauli operators that we need for the proofs in this article.
\begin{lemma}\label{sum-Pauli}
We have for all $Q\in \{\mathds{I},X,Y,Z\}^{\otimes n}$:
\begin{align*}
    \sum_{P\in \mathds{P}_n}  (-1)^{P{\circ} Q} = d^2\cdot  \mathbb{1}_{Q=\mathds{I}}.
\end{align*}
\end{lemma}
\begin{proof}
It is clear that for $Q=\mathds{I}$, $Q$ commutes with every $P\in \mathds{P}_n$ and thus the equality holds. Now, let $Q\in \mathds{P}_n\setminus \{\mathds{I}\}$ and we write $Q=Q_1\otimes \dots \otimes Q_n$ where for all $i\in [n]$, $Q_i\in\{\mathds{I},X,Y,Z\} $ is a Pauli matrix. By the same decomposition for $P\in \mathds{P}_n$, we can write:
\begin{align*}
    \sum_{P\in \mathds{P}_n}  (-1)^{P{\circ} Q} &=  \sum_{P_1,\dots,P_{n}\in \{\mathds{I},X,Y,Z\}} (-1)^{P_1{\circ}Q_1+_2\dots+_2 P_n{\circ}Q_n }
    \\&= \prod_{i=1}^n \sum_{P_i\in \{\mathds{I},X,Y,Z\}} (-1)^{P_i{\circ}Q_i }
    \\&= \prod_{i=1}^n 4\mathbb{1}_{Q_i=\mathds{I}_2}
    \\&= d^2\mathbb{1}_{Q=\mathds{I}_d}
\end{align*}
where we have used in the third equality the fact that every non identity Pauli matrix $Q_i$ commutes only with the identity and itself (so it anti-commutes with the two other Pauli matrices) thus the sum $\sum_{P_i\in \{\mathds{I},X,Y,Z\}} (-1)^{P_i{\circ}Q_i }=0$.
\end{proof}
\begin{lemma}\label{int-Pauli}
We have for all matrices $\rho$:
\begin{align*}
    \sum_{P\in \{\mathds{I},X,Y,Z\}^{\otimes n}} P\rho P =  d\tr(\rho)\mathds{I}.
\end{align*}
\end{lemma}
\begin{proof}
Let $d=2^n$ and $\rho\in \mathds{C}^{d\times d}$. It is known that $\frac{1}{\sqrt{d}}\{\mathds{I},X,Y,Z\}^{\otimes n}$ forms an ortho-normal basis of $\mathds{C}^{d\times d}$ for the Hilbert-Schmidt {inner} product. Thus, we can write $\rho$ in this basis:
\begin{align*}
    \rho\!= \! \sum_{P\in \{\mathds{I},X,Y,Z\}^{\otimes n}}\! \tr\left( \!\frac{P}{\sqrt{d}} \rho \!\right) \!\frac{P}{\sqrt{d}}\!=\! \frac{1}{d}\!\sum_{P\in \{\mathds{I},X,Y,Z\}^{\otimes n}} \!\tr\left( P\rho \right)\! P. 
\end{align*}
Therefore we can simplify the LHS by using the identity $PQ=(-1)^{P{\circ} Q}QP$ for all $P,Q\in \mathds{P}_n$:
\begin{align*}
        \sum_{P\in \mathds{P}_n} P\rho P &= \frac{1}{d}\sum_{P,Q \in \mathds{P}_n} \tr(Q\rho) PQP 
       \\& = \frac{1}{d}\sum_{P,Q \in \mathds{P}_n} \tr(Q\rho) (-1)^{P{\circ} Q}QPP  
        \\&=  \sum_{Q \in \mathds{P}_n} \tr(Q\rho) Q \frac{1}{d}\sum_{P\in \mathds{P}_n}  (-1)^{P{\circ} Q} 
       \\& =  \sum_{Q \in \mathds{P}_n} \tr(Q\rho) Q \cdot d\cdot\mathbb{1}_{Q=\mathds{I}}
        \\&= d\tr(\rho)\mathds{I},
\end{align*}
where we have used Lemma~\ref{sum-Pauli} to obtain the fourth equality.
\end{proof}
\subsection{Kirszbraun theorem}\label{Kirszbraun theorem}
\begin{theorem}[Kirszbraun, \cite{mattila1999geometry}]
    If $U$ is a subset of $\mathds{R}$ and $f:U\rightarrow \mathds{R}$ is an $L$ Lipschitz function with respect to a distance $\mathbf{d}$, then  there is a Lipschitz function $g: \mathds{R}\rightarrow \mathds{R}$ that extends $f$ and has the same Lipschitz constant $L$ as $f$ with respect to the distance $\mathbf{d}$. 
    Moreover, the extension is provided by 
    \begin{align*}
        g(x) = \inf_{y\in U} \left(f(y) +L\cdot \mathbf{d}(x,y) \right).
    \end{align*}
\end{theorem}
\subsection{Concentration of Lipschitz functions of Gaussian random variables}\label{concentration}
\begin{theorem}[\cite{wainwright2019high}, Theorem 2.26]
    Let $(X_1, \dots, X_n)$ be a vector of i.i.d. standard Gaussian variables, and let $f:\mathds{R}^{n}\rightarrow \mathds{R}$ be $L$-Lipschitz with respect to the Euclidean norm. Then we have for all $t\ge 0$:
    \begin{align*}
        \pr{|f(X)-\ex{f(X)}|\ge t } \le 2e^{-\frac{t^2}{2L^2}}.
    \end{align*}
\end{theorem}

\subsection{Gaussian integration by parts }
Gaussian integration by parts (see e.g. \cite{van2014probability}) is a generalization of Isserlis' formula \cite{isserlis1918formula}.
\begin{theorem}\label{thm:GIBP}
Let $(X_1,\dots,X_d)$ be a Gaussian vector and $f:\mathds{R}^d\rightarrow \mathds{R}$  be a smooth function. We have:
\begin{align*}
    \ex{X_1f(X_1,\dots,X_d)}\!=\!\sum_{i=1}^d\! \cov(X_1,X_i) \ex{\partial_if(X_1,\dots,X_d)}.
\end{align*}
\end{theorem}

\section{Proof of \eqref{equ:approximation_mutual_mcdiarmid}}
In this section, we give the proof of an approximation used in the proof of Theorem~\ref{thm: LB-adaptive}, more precisely in \eqref{equ:approximation_mutual_mcdiarmid}, where we showed that the empirical average over the ensemble of a certain function is well-approximated by its mean.
\begin{proposition}\label{prop:approx}
There is a universal constant $C>0$ such that with probability at least $9/10$ we have:
    \begin{align*}
        & \sum_{k=1}^N\frac{1}{M}\sum_{x=1}^M\sum_{i_1,\dots,i_{k-1}} \left(\prod_{t=1}^{k-1}\lambda_{i_t }^t\left(\frac{1+u_{i_t}^{t,x}}{d}\right)\right)\sum_{i_k}\frac{\lambda_{i_k}^k}{d}(u_{i_k}^{k,x})^2
        \\&\le \sum_{k=1}^N\mathds{E}_\alpha\Bigg[\!\sum_{i_1,\dots,i_{k-1}}\! \left(\prod_{t=1}^{k-1}\lambda_{i_t }^t\left(\frac{1+u_{i_t}^{t,\alpha}}{d}\right)\!\right)\sum_{i_k}\frac{\lambda_{i_k}^k}{d}(u_{i_k}^{k,\alpha})^2\Bigg] +N\eps^2\exp(-Cd^2).
    \end{align*}
\end{proposition}
\begin{proof}Let $k\in [N]$. For $x\in [M]$, let $f_k(x)$ be the function:
    \begin{align*}
        f_k(x)=\sum_{i_1,\dots,i_{k-1}} \left(\prod_{t=1}^{k-1}\lambda_{i_t }^t\left(\frac{1+u_{i_t}^{t,x}}{d}\right)\right)\sum_{i_k}\frac{\lambda_{i_k}^k}{d}(u_{i_k}^{k,x})^2.
    \end{align*}
    {where we recall that the random variable $u_{i_t}^{t,x}$ is defined as \begin{align*}
        u_{i_t}^{t,x}=\sum_{P\in \mathds{P}_n} \frac{2\alpha_x(P)\eps}{\|\alpha_x\|_2}\bra{\phi_{i_k}^k} P\rho_kP\ket{\phi_{i_k}^k}-\sum_{P\in \mathds{P}_n} \frac{2\alpha_x(P)\eps}{d\|\alpha_x\|_2}.
    \end{align*}
    where $\alpha_x(P)\sim \cN(0,1)$.}
    Similarly we define {for $\alpha=\left(\alpha(P) \right)_{P\in \mathds{P}_n}$ and $\alpha(P)\sim \cN(0,1)$:}
    \begin{align*}
        f_k(\alpha) \coloneqq  \sum_{i_1,\dots,i_{k-1}} \left(\prod_{t=1}^{k-1}\lambda_{i_t }^t\left(\frac{1+u_{i_t}^{t,\alpha}}{d}\right)\right)\sum_{i_k}\frac{\lambda_{i_k}^k}{d}(u_{i_k}^{k,\alpha})^2.
    \end{align*}
    {Now define the random variables 
    \begin{align*}
        F(x)=\sum_{k=1}^N f_k(x) \quad \text{and } \quad F(\alpha)= \sum_{k=1}^N f_k(\alpha).
    \end{align*}
     We want to show a concentration of the random variable $\frac{1}{M}\sum_{x=1}^M  F(x)$ around its mean $\mathds{E}_{\alpha}(F(\alpha))$. If the random variables $\left(F(x)\right)_x$ are bounded we can  use Hoeffding's inequality to obtain a concentration inequality for the empirical mean $\frac{1}{M}\sum_{x=1}^M  F(x)$.}

     {From Lemma~\ref{lem: useful}, we have for all $x\in [M]$, for all $k\in [N]$:
     \begin{align*}
         \sum_{i_k}\frac{\lambda_{i_k}^k}{d}(u_{i_k}^{k,x})^2\le 16\eps^2
     \end{align*}
     Hence using $\sum_{i_k}\frac{\lambda_{i_k}^k u_{i_k}^{k,x}}{d}=0$ and $\sum_{i_k}\frac{\lambda_{i_k}^k}{d}=1$ we obtain:
     \begin{align*}
        f_k(x)&=\sum_{i_1,\dots,i_{k-1}} \left(\prod_{t=1}^{k-1}\lambda_{i_t }^t\left(\frac{1+u_{i_t}^{t,x}}{d}\right)\right)\sum_{i_k}\frac{\lambda_{i_k}^k}{d}(u_{i_k}^{k,x})^2
        \\&\le \sum_{i_1,\dots,i_{k-1}} \left(\prod_{t=1}^{k-1}\lambda_{i_t }^t\left(\frac{1+u_{i_t}^{t,x}}{d}\right)\right)16\eps^2
        = 16\eps^2
    \end{align*}
    which implies an upper bound on the random variable $F$: 
     \begin{align*}
        0\le F(x)=\sum_{k=1}^N f_k(x) \le 16N \eps^2. 
    \end{align*}
    Therefore Hoeffding's inequality \cite{hoeff} implies:
    \begin{align*}
       & \pr{\left|\frac{1}{M}\sum_{x=1}^M F(x)-\ex{\frac{1}{M}\sum_{x=1}^M F(x)}\right|>s}
        \le 2\exp\left(-\frac{2s^2M}{16^2N^2\eps^4}\right).
    \end{align*}
     }
    Since for all $x\in [M]$ we have $\mathds{E}_{\alpha_x}(f_k(x))=\mathds{E}_\alpha (f_k(\alpha)) $, we deduce:
     \begin{align*}
        &\pr{\left| \sum_{k=1}^N\frac{1}{M}\sum_{x=1}^M f_k(x)-\sum_{k=1}^N\mathds{E}_\alpha (f_k(\alpha)) \right|>s}
        \le 2\exp\left(-\frac{s^2M}{128 N^2\eps^4}\right).
    \end{align*}
    Finally, by taking $s=12N\eps^2\sqrt{\frac{\log(20)}{M}}$, with probability at least $9/10$, we have: 
    \begin{align*}
        &\sum_{k=1}^N\frac{1}{M}\sum_{x=1}^M\sum_{i_1,\dots,i_{k-1}} \left(\prod_{t=1}^{k-1}\lambda_{i_t }^t\left(\frac{1+u_{i_t}^{t,x}}{d}\right)\right)\sum_{i_k}\frac{\lambda_{i_k}^k}{d}(u_{i_k}^{k,x})^2
        \\&= \sum_{k=1}^N\frac{1}{M}\sum_{x=1}^M f_k(x)
        \le  \sum_{k=1}^N\mathds{E}_\alpha (f_k(\alpha)) + 12N\eps^2\sqrt{\frac{\log(20)}{M}}
        \\&\le \sum_{k=1}^N \mathds{E}_\alpha\Bigg[\!\sum_{i_1,\dots,i_{k-1}} \!\left(\prod_{t=1}^{k-1}\lambda_{i_t }^t\left(\frac{1+u_{i_t}^{t,\alpha}}{d}\right)\!\right)\sum_{i_k}\frac{\lambda_{i_k}^k}{d}(u_{i_k}^{k,\alpha})^2\Bigg] 
      +N\eps^2\exp(-Cd^2)
    \end{align*}
    where $C>0$ is a universal constant and we used the fact that $M=\exp(\Omega(d^2))$.
    \end{proof}

\end{document}